\documentclass[10pt, journal,compsoc]{IEEEtran}
\usepackage{bbm}
\usepackage{amssymb}
\usepackage{amsthm}
\usepackage{mathtools}
\usepackage{graphicx}
\usepackage{epstopdf}
\DeclareGraphicsExtensions{.eps}
\usepackage{algorithm, amsmath}
\usepackage{algpseudocode}
\usepackage{float}
\usepackage{subcaption}
\usepackage{etoolbox}
\usepackage[T1]{fontenc}
\usepackage{multirow}
\usepackage{enumitem}
\usepackage{diagbox}
\usepackage{cite}
\DeclareMathOperator*{\argmin}{arg\,min}
\DeclareMathOperator*{\argmax}{arg\,max}

\newcommand{\V}{\mathcal{V}}
\newcommand{\E}{\mathcal{E}}
\newcommand{\F}{\mathcal{F}}
\newcommand{\U}{\mathcal{U}}

\newcommand{\R}{\mathcal{R}}

\newcommand{\HH}{\mathcal{H}}

\newcommand{\blue}{\textcolor{black}}

\allowdisplaybreaks

\usepackage{tikz}
\usetikzlibrary{arrows.meta}
\tikzset{>={Latex[width=2mm,length=2mm]}}

\tikzstyle{vertex}=[circle, draw]
\newtheorem{theorem}{Theorem}
\newtheorem{lemma}{Lemma}

\tikzstyle{stuff_fill}=[vertex, style=green, fill=black!10]
\usetikzlibrary{calc}

\newcounter{defcounter}
\newcounter{mycounter}
\IEEEoverridecommandlockouts

\begin{document}
	
    \title{Placement and Allocation of Virtual Network Functions: Multi-dimensional Case}
	
	\author{Gamal Sallam, Zizhan Zheng, and Bo Ji \thanks{This work was supported in part by the NSF under Grants CNS-1651947. A preliminary version of this work was presented at IEEE ICNP 2019 \cite{sallam2019placement}.
		
		Gamal Sallam (tug43066@temple.edu) is with the Department of
		Computer and Information Sciences, Temple University, Philadelphia, PA,  Zizhan Zheng (zzheng3@tulane.edu) is	with the
    	Department of Computer Science, Tulane University, New Orleans, LA, and Bo Ji (boji@vt.edu) is with the Department of Computer Science, Virginia Tech, Blacksburg, VA. Bo Ji is the corresponding author.
}}
	
	\maketitle

	\IEEEpeerreviewmaketitle
	\begin{abstract}
	Network function virtualization (NFV) is an emerging design paradigm that replaces physical middlebox devices with software modules running on general purpose commodity servers. While gradually transitioning to NFV, Internet service providers face the problem of where to introduce NFV in order to make the most benefit of that; here, we measure the benefit by the amount of traffic that can be served in an NFV-enabled network. This problem is non-trivial as it is composed of two challenging subproblems: 1) placement of nodes to support virtual network functions (referred to as VNF-nodes); 2) allocation of the VNF-nodes' resources to network flows. These two subproblems must be jointly considered to satisfy the objective of serving the maximum amount of traffic.  This problem has been studied for the one-dimensional setting, where all network flows require one network function, which requires a unit of resource to process a unit of flow. In this work, we consider the multi-dimensional setting, where flows must be processed by multiple network functions, which require a different amount of each resource to process a unit of flow. The multi-dimensional setting introduces new challenges in addition to those of the one-dimensional setting (e.g., NP-hardness and non-submodularity) and also makes the resource allocation subproblem a multi-dimensional generalization of the generalized assignment problem with assignment restrictions. To address these difficulties, we propose a novel two-level relaxation method that allows us to draw a connection to the sequence submodular theory and utilize the property of sequence submodularity along with the primal-dual technique to design two approximation algorithms. We further prove that the proposed algorithms have a non-trivial approximation ratio that depends on the number of VNF-nodes, resources, and a measure of the available resource compared to flow demand. Finally, we perform 
	trace-driven simulations to show the effectiveness of the proposed algorithms.
	\end{abstract}

	\section{Introduction}
	
	With the emergence of many new technologies and applications (such as autonomous vehicles, extended reality, and edge intelligence), the role of modern networks (beyond-5G and 6G) has evolved beyond providing basic connectivity services \cite{latva2020key}. It has now 
become imperative to provide various types of network services (such as security, performance
optimization, and value-added services) in modern networks.
Realizing such network services usually requires a new design paradigm, called Network function virtualization (NFV), where network functions (e.g., firewall, intrusion detection, and load balancer) that traditionally run in dedicated hardware are now replaced by software modules hosted on general purpose commodity servers \cite{Introduction2012}. Several advantages can be harnessed from this architecture such as reducing the deployment cost, increasing the agility, and improving the scalability. These advantages have encouraged major Internet service providers (ISPs) to consider this new architecture, and some of them have already started the transition to NFV \cite{Amdocs_whitepaper}.

	
	
	However, transitioning to NFV faces challenges from different perspectives. From network flows' perspective, each flow needs to be processed by certain types of network functions, and each network function requires a different amount of the resources at servers (e.g., CPU, memory, and I/O). In addition, flows generally require all of their traffic be fully processed by such functions to satisfy certain quality of services \cite{poularakis2017one}. From ISPs' perspective, transitioning to NFV usually happens in multiple stages for several reasons, including budget limitations and the desire to utilize the already provisioned hardware. Considering both perspectives leads to an important question: under a limited budget, how to efficiently introduce NFV in each stage such that the total traffic of fully processed flows is maximized? To answer this question, we need to address two main issues: 1) where to place nodes that support NFV (called VNF-nodes) without exceeding the given budget? And 2) how to allocate the VNF-nodes' resources to satisfy the requirements of network flows? We refer to this problem as joint VNF-nodes placement and resource allocation (VPRA). 
	
	Most of the previous work either does not consider a limited budget (e.g., \cite{sang2017provably}) or relaxes the resources constraint (e.g., \cite{poularakis2017one}). In \cite{sallam2021joint}, both the budget and resources constraints are considered, along with the requirement that flows must be fully processed. However, they consider a special case of the VPRA problem with the following characteristics: a) there is only one type of resource; b) all flows require the same network function; c) the network function requires one unit of resource to process each unit of flows (we refer to this setting as basic-VPRA). Even under such a simplified setting, the basic-VPRA is already quite challenging. It is shown in  \cite{sallam2021joint} that not only is this problem NP-hard, but it does not possess a useful property, called submodularity, which oftentimes leads to efficient solutions. In this work, we take one step further and extend the basic-VPRA problem to the setting with multiple network functions, multiple resources, and heterogeneous resource requirements. We refer to this generalization as multi-dimensional VPRA (multi-VPRA). 
	
	We systematically study the challenges of the multi-VPRA problem and show that the difficulties introduced by the generalization call for different design strategies and analytical techniques. Specifically, we decompose the original problem into two subproblems: placement and resource allocation. We show that the placement subproblem even without the requirement that flows must be fully processed is NP-hard. Moreover, the resource allocation subproblem is a multi-dimensional generalization of the generalized assignment problem with assignment restrictions, which is also NP-hard \cite{chekuri1999multi}. To address the placement subproblem, we introduce a novel two-level relaxation method that allows us to draw a connection to the sequence submodular (also called string submodular) theory \cite{zhang2016string, alaei2010maximizing} and design an efficient placement algorithm. 
	\blue{Note that sequence submodularity is a generalization of submodularity. Like submodular functions, sequence submodular functions also exhibit the diminishing returns property, meaning that the value of adding an item to a sequence decreases as the sequence expands. For sequence functions, forward (resp., backward) monotonicity means the value of the function increases when an item is added to the end (resp., beginning) of the sequence. We provide formal definitions of these properties in Section~\ref{subsec:sequence_submodular}.}
	For the resource allocation subproblem, we utilize the primal-dual technique \cite{williamson2011design} to design two efficient resource allocation algorithms. We combine the placement algorithm with the resource allocation algorithms and develop approximation algorithms with performance guarantees for the original non-relaxed multi-VPRA problem.
	
    Our main contributions are summarized as follows.
    \begin{itemize}
        \item First, we systematically study the challenges arising from the generalized multi-VPRA problem. In addition to the challenges faced by the basic-VPRA (such as NP-hardness and non-submodularity), we show that overcoming the non-submodularity of the placement subproblem is much harder and that the resource allocation subproblem is a multi-dimensional generalization of the generalized assignment problem with assignment restrictions, which is also more challenging.
        \item Second, we introduce a novel two-level relaxation method that enables us to convert the non-submodular placement subproblem into a sequence submodular optimization problem. In order to leverage the property of sequence submodularity, we generalize the concept of backward-monotone to approximate backward-monotone, extend the known results for backward-monotone to this generalized version, and utilize this new property to develop an efficient algorithm for the placement of VNF-nodes.
        \item Third, we utilize the primal-dual technique to design two efficient resource allocation algorithms. Moreover, we show that by combining the proposed placement algorithm and the two resource allocation algorithms, we can achieve an approximation ratio of $\frac{(e-1)(Z-1) } {4 e^2 Z (kR)^{1/(Z-1)}}$ and $\frac{(e-1)(Z-1)}{4e(Z-1+e Z R^{1/(Z-1)})}$ for the original non-relaxed multi-VPRA problem, respectively, where $k$ (resp. $R$) is the number of VNF-nodes (resp. resources), and $Z$ is a measure of the available resource compared to flow demand. When $Z$ goes to infinity, the approximation ratios become constants: $\frac{e-1}{4e^2}$ and $\frac{e-1}{4e^2+4e}$, respectively.  
        \item Finally, we conduct trace-driven simulations using Abilene dataset \cite{abilene} as well as datasets from SNDlib \cite{orlowski2010sndlib} to evaluate the performance of the proposed algorithms. 
    \end{itemize}
    
    \section{System model}

    \begin{table}[t]
        \centering
        \begin{tabular}{c|c}
            \hline
            \textbf{Symbol} & \textbf{Definition} \\  \hline \hline
            
            $\V$ & The set of nodes in the network\\  \hline
            $\E$ & The set of edges connecting the nodes in the network\\  \hline
            $\F$ & The set of flows\\  \hline
            $\lambda_f$ & The traffic rate of flow $f$ \\  \hline
            $\boldsymbol{\lambda}$ & The flow rate vector \\  \hline
            $\V_f$ & The set of nodes along the path of flow $f$ \\   \hline
            \begin{tabular}{l}
                $\F_\U$
          \end{tabular} &
          \begin{tabular}{l}
            The set of all flows whose path has at least \\ one node in a subset of nodes $\mathcal{U}$
          \end{tabular} \\ \hline
            $\Phi$ & The set of network functions \\   \hline
            $\Phi_f$ & The set of network functions required by flow $f$\\   \hline
            $\F(\phi)$ & The set of flows that require network function $\phi$ \\   \hline
            $\mathcal{R}$ & The set of resource types \\   \hline
            $c_v^r$ & The total amount of resource $r$ available at node $v$ \\   \hline
            \begin{tabular}{l} $\beta_\phi^r$ \end{tabular} & \begin{tabular}{l} The amount of resource $r$ needed for function \\ $\phi$ to process one unit of a network flow \end{tabular} \\ \hline
            $x_f^v$ & The portion of flow $f$ assigned to VNF-node $v$ \\   \hline
            $\boldsymbol{X}(\U)$ & The flow assignment matrix of a set of nodes $\U$ \\  \hline
        \end{tabular}
         \caption{Summary of notations}
        \label{table:symbols}
    \end{table}

We consider a network graph $G=(\V, \E)$, where $\V$ is the set of nodes, with $V = |\V|$, and $\E$ is the set of edges. We have a set of flows $\F$, with $F = |\F|$. We use $\lambda_f$ to denote the traffic rate of flow $f \in \F$. Let $\boldsymbol{\lambda} \triangleq [\lambda_{f_1}, \dots, \lambda_{f_F}]$ be the flow rate vector. As we mention earlier, the transition to NFV happens in two phases: the planning phase and the production phase. In this work we focus on the planning phase and assume that the traffic of flow $f$ will be sent along a predetermined path (e.g., a shortest path), and the set of nodes along this path is denoted by $\V_f$. We use $\F_\U$ to denote the set of all flows whose path has at least one node in a subset of nodes $\mathcal{U} \subseteq \mathcal{V}$, i.e., $\F_{\U} = \{f \in \F ~|~ \V_f \cap \U \neq \emptyset\}$. When a node can support some VNFs, we call it a VNF-node. Since ISPs have a limited budget to deploy VNFs in their networks, they can only choose a subset of nodes $\U \subseteq \V$ to become VNF-nodes. 
    
    We consider a set of network functions denoted by $\Phi$. Each flow needs to be processed by one or more network functions. We assume that flows can get processed by all required network functions at the same VNF-node. This has the potential of reducing the overhead of maintaining flow state across VNF-nodes and can be realized by the NFV architecture \cite{Introduction2012}, which allows hosting different types of network functions at the same VNF-node. The set of network functions required by flow $f$ is denoted by $\Phi_f$. The set of flows that require network function $\phi \in \Phi$ is denoted as $\F(\phi)$. Each VNF-node $v \in \V$ can host one or more network functions. We use $\mathcal{R}$ to denote the set of resource types at VNF-nodes (e.g., memory, CPU, and I/O), with $R = |\R|$. Each network function $\phi$ requires $\beta_\phi^r$ units of resource $r \in \R$ to process one unit of a network flow. The traffic rate $\lambda_f$ of each flow can be split and can be processed at multiple VNF-nodes. We use $x_f^v$ to denote the portion of flow $f$ that is assigned to VNF-node $v$ and use $\boldsymbol{X}(\V) \in \mathcal{R}^{F \times V}$ to denote the assignment matrix. 
    \blue{In Table~\ref{table:symbols}, we summarize the main notations that will be used in the problem formulation.} 
    
    As we mentioned earlier, the benefits of processed traffic can be harnessed from fully processed flows, i.e., flows that have all of their traffic fully processed at VNF-nodes. Hence, when a flow traverses VNF-nodes and there are sufficient resources on these VNF-nodes to process all of its rate, i.e., $\sum_{v \in \V_f \cap \U} x_f^v \geq \lambda_f$, then the flow is counted as a processed flow. Therefore, the total fully processed traffic for a subset of VNF-nodes $\U \subseteq \V$ can be expressed as follows:
    \begin{equation}
    \label{eq:objectiveJ_1}
    {J_1}(\U, \boldsymbol{X}(\U)) \triangleq \sum_{f\in \F} \lambda_f \boldsymbol{1}_{\{ \sum_{v \in \V_f \cap \U} x_f^v \geq \lambda_f \}},
    \end{equation}
    where $\boldsymbol{1}_{\{\cdot\}}$ is the indicator function. However, there is a total amount of each resource available at the nodes, and the amounts could be different at different nodes. We use $c_v^r$ to denote the total amount of resource $r$ at node $v$. Then, the following constraints should be satisfied:
    \begin{equation}
    \begin{cases}
    \sum_{\phi \in \Phi} \beta_\phi^r \sum_{f \in \F(\phi)} x_f^v \leq c_v^r, \, & \forall  r \in \R \text{ and } v \in \U \label{eq:nodecapacity},\\
    x_f^v = 0,              & \forall f \in \F \text{ and } \forall v \notin \U.
    \end{cases}
    \end{equation}
    Also, we consider a limited budget $B$ and assume that the cost for making node $v$ a VNF-node is the same for all nodes, which is denoted by $b$.  Let $k=\lfloor B/b \rfloor$. Then, the budget constraint can be expressed as a cardinality constraint, i.e.,
    \begin{equation}
    \label{eq:budget}
    |\U| \leq k.
    \end{equation}

    As a service provider with a limited budget, a plausible objective is to introduce NFV at nodes that would result in the maximum fully processed traffic. Therefore, we consider the problem of multi-dimensional VNF-nodes placement and resource allocation (multi-VPRA) with the objective of maximizing the total fully processed traffic ($J_1(\U, \boldsymbol{X}(\U))$). The problem can be formulated as follows:
    \begin{equation}\tag{$P1$}
    \label{eq:mainProblem}
    \begin{aligned}
    & \underset{\U \subseteq \V, \boldsymbol{X}(\U)}{\text{maximize}} \quad  J_1(\U, \boldsymbol{X}(\U))\\
    & \text{subject to} \quad \eqref{eq:nodecapacity} \text{ and }\eqref{eq:budget}.
    \end{aligned}
    \end{equation}
    
    \section{Challenges of Multi-VPRA}
    In this section, we analyze the multi-VPRA problem and identify the main challenges posed by this problem. We first decompose the multi-VPRA problem into two subproblems: 1) placement, i.e., where to deploy VNF-nodes; 2) resource allocation of the VNF-nodes among flows. We will show the hardness of each subproblem and explain new challenges arising from the multi-dimensional generalization.
    
    \subsection{Decomposition}
    In this subsection, we present a decomposition of the multi-VPRA problem into placement and allocation subproblems. We start with the allocation subproblem because it will be used in the placement subproblem. For a given set of VNF-nodes $\U \subseteq \V$, let $J_2(\boldsymbol{X}(\U))$ denote the total amount of fully processed traffic under flow assignment $\boldsymbol{X}(\U)$. Then, the resource allocation subproblem for a given set of VNF-nodes $\U$ can be formulated as follows:
    \begin{equation}\tag{$P2$}
    \label{eq:allocation}
    \begin{aligned}
    & \underset{\boldsymbol{X}(\U): \eqref{eq:nodecapacity}~\text{is satisfied}} {\text{maximize}} \quad  {J_2}( \boldsymbol{X}(\U)). 
    \end{aligned}
    \end{equation}
    Let $J_3(\U) \triangleq \max_{\boldsymbol{X}(\U): \eqref{eq:nodecapacity}~\text{is satisfied}} J_2 (\boldsymbol{X}(\U))$ denote the placement value function, which is the optimal value of Problem \eqref{eq:allocation} for a given set of VNF-nodes $\U$. Then, the placement subproblem can be formulated as follows:
    \begin{equation} \tag{$P3$}
    \label{eq:placement}
    \begin{aligned}
    & \underset{\U \subseteq \V}{\text{maximize}} \quad   {J_3}( \U)\\
    & \text{subject to} \quad \eqref{eq:budget}. 
    \end{aligned}
    \end{equation}
    Note that in order to solve subproblem \eqref{eq:placement}, we need to solve subproblem \eqref{eq:allocation} first to find the optimal $\boldsymbol{X}(\U)$ for a given set of VNF-nodes  $\U$.
    \subsection{ Hardness}
    In \cite[Theorem 1]{sallam2021joint}, it is shown that for the basic-VPRA problem, both subproblems \eqref{eq:allocation} and \eqref{eq:placement} are NP-hard. The NP-hardness results can be easily extended to the multi-dimensional case considered here. Therefore, we simply state the hardness results in the following lemma without proofs.
    \begin{lemma}
        \label{lemma:nphardness}
        The resource allocation subproblem \eqref{eq:allocation} and the placement subproblem \eqref{eq:placement} are both NP-hard.    
    \end{lemma}
    
    In addition, the placement subproblem of the basic-VPRA has been shown to be non-submodular \cite[Section IV. B]{sallam2021joint}. Similarly, the non-submodularity result can also be easily extended to the multi-dimensional case. In order to develop efficient algorithms for the basic-VPRA, \blue{the work of \cite{sallam2021joint} employs a simple relaxation method that allows any portion of a partially processed flow to be counted in the objective function. That is, the objective function in Eq.~\eqref{eq:objectiveJ_1} becomes $\sum_{f\in \F} \sum_{v \in \V_f \cap \U} x_f^v$. Such a simple relaxation allows one to prove submodularity of the placement subproblem and to design efficient algorithms for the basic-VPRA.} However, in the sequel, we will explain why the same framework and algorithms cannot be directly applied to solve the multi-VPRA problem we consider.
    
    The first challenge is that a similar relaxation of the basic-VPRA does not admit an efficient placement algorithm with performance guarantees for the multi-VPRA problem. The reason is that the objective function of the relaxed placement subproblem of the basic-VPRA problem can be shown to be equivalent to the maximum flow problem, which can be proved to be submodular. In contrast, the objective function of the relaxed placement subproblem of the multi-VPRA problem, to the best of our knowledge, can only be evaluated using Linear Programming, which does not provide us with enough insights that can be utilized to prove or disprove submodularity. The second challenge is that the resource allocation algorithms proposed for the basic-VPRA consider only a single resource and cannot be utilized to provide performance guarantees for the multi-VPRA problem, where multiple resources have to be considered during the resource allocation.
    
    In order to address these new challenges, we introduce a novel two-level relaxation method: (i) we allow partially processed flows as in \cite{sallam2021joint}, and (ii) we consider an approximate version of the resource allocation subproblem. This new relaxation method enables us to make a connection between the relaxed placement subproblem and the sequence submodular theory and design an efficient placement algorithm. For the resource allocation, we design two resource allocation algorithms both based on the primal-dual technique. Not only the proposed placement and resource allocation algorithms can properly handle the multi-dimensional setting, but they also guarantee a constant approximation ratio for the original non-relaxed multi-VPRA problem. 
    
    \section{Relaxed Multi-VPRA}
        \label{subsec:relaxedFormulation}
    In this section, we present the two-level relaxation of the multi-VPRA problem. In the first-level, we allow partially processed flows to be counted in the objective function, and in this case we use $R_1(\U, \boldsymbol{X}(\U))$  to denote the relaxed objective function (defined in Eq.  \eqref{eq:relaxedObjective}). In the second-level, instead of evaluating function $R_1(\U, \boldsymbol{X}(\U))$  for a set of nodes $\U$ together, we allow the algorithm to consider a specific ordering of nodes and evaluate the objective function on a node-by-node basis. Apparently, the first-level relaxation does not decrease the total traffic that can be assigned to a given set of VNF-nodes $\U$. In contrast, the second-level relaxation results in an approximate version of the resource allocation subproblem, and thus, there is a loss in the amount of processed traffic. However, we will prove that the loss is at most $1/2$ of the optimal. In addition, through simulation results, we will show that the loss due to the second-level relaxation is negligible. The purpose of this two-level relaxation is to draw a connection to the sequence submodular theory, which enables us to design efficient algorithms with provable performance guarantees.

    \subsection{First-level Relaxation}
    We first introduce the first-level relaxation, which allows partially processed flows to be counted. In this case, any fraction of  flow $f$ processed by VNF-nodes in $\V_f \cap \U$ will be counted in the total processed traffic. That is, the relaxed $J_1(\U, \boldsymbol{X}(\U))$ can be expressed as follows:
    \begin{equation}
    \label{eq:relaxedObjective}
    R_1(\U, \boldsymbol{X}(\U)) \triangleq \sum_{f\in \F} \sum_{v\in \V_f \cap \U} x_f^v.
    \end{equation}
    Apparently, the total processed traffic of flow $f$ cannot exceed $\lambda_f$, i.e.,  the flow rate constraint needs to be satisfied:
    \begin{equation}
    \sum_{v \in \U} x_f^v \leq \lambda_f, \quad \forall f \in \F. \label{eq:traffic2}
    \end{equation}
    Then, after the first-level relaxation, Problem \eqref{eq:mainProblem} becomes
    \begin{equation}\tag{$Q1$}
    \label{eq:relaxedProblem}
    \begin{aligned}
    \underset{\U \subseteq \V, \boldsymbol{X}(\U)}{\text{maximize}} \quad & R_1(\U, \boldsymbol{X}(\U))\\
    \text{subject to} \quad & \eqref{eq:nodecapacity}, \eqref{eq:budget},~\text{and}~ \eqref{eq:traffic2}. 
    \end{aligned}
    \end{equation}
   
    Next, we explain why we need the second-level relaxation for solving the multi-VPRA problem efficiently. Similar to the decomposition of Problem \eqref{eq:mainProblem}, we also decompose Problem \eqref{eq:relaxedProblem} into placement and allocation subproblems. For a given set of VNF-nodes $\U \subseteq \V$, let $\mathcal{X}^{\U}$ be the set of all flow assignment matrices $\boldsymbol{X}(\U)$ that satisfy the resources constraint \eqref{eq:nodecapacity} and the flow rate constraint \eqref{eq:traffic2}, and let ${R_2}(\boldsymbol{X}(\U))$ be the total processed traffic. Then, the resource allocation subproblem for a given set of VNF-nodes $\U$ can be formulated as
    \begin{equation}\tag{$Q2$}
    \label{eq:relaxedAllocation}
    \begin{aligned}
    & \underset{\boldsymbol{X}(\U) \in \mathcal{X}^{\U}} {\text{maximize}} \quad  {R_2}( \boldsymbol{X}(\U)).
    \end{aligned}
    \end{equation}
   
    Now, let ${R_3}(\U) \triangleq \max_{\boldsymbol{X}(\U) \in \mathcal{X}^{\U}} {R_2}(\boldsymbol{X}(\U))$ denote the optimal value of Problem \eqref{eq:relaxedAllocation} for a given set of VNF-nodes $\U$. The function $R_3(\U)$ is also called the placement value function, and the placement subproblem can be formulated as
    \begin{equation}\tag{$Q3$}
    \label{eq:relaxedPlacement}
    \begin{aligned}
    & \underset{\U \subseteq \V}{\text{maximize}} \quad   {R_3}( \U)\\
    & \text{subject to} \quad \eqref{eq:budget}. 
    \end{aligned}
    \end{equation}
    Unlike the relaxed placement subproblem of the basic-VPRA problem, which has been proven to be submodular, the submodularity of the relaxed placement subproblem \eqref{eq:relaxedPlacement} of the multi-VPRA remains unknown as explained earlier. Driven by this observation, in the next subsection we introduce another level of relaxation, which enables us to draw a connection to the sequence submodular theory. 
    
    \subsection{Second-level Relaxation}
    \label{subsec:second_relaxation}
In the second-level relaxation, instead of solving subproblem \eqref{eq:relaxedAllocation} to obtain the optimal solution $R_3(\U)$ for a set of nodes $\U$, we consider a specific ordering of nodes $\U$ and solve for each node one-by-one according to their order (which will be explained soon in Algorithm \ref{alg:fractional}). By doing this, we make a connection to the sequence submodular theory, which enables us to design an efficient placement algorithm with provable performance guarantee. First, we give some additional notations. Let $(v_1, v_2, \dots, v_m)$ be a sequence of nodes selected over $m$ steps, where $v_i \in \V$ is selected in the $i$-th step for $i = 1, \dots, m$. Let the set of all possible sequences of nodes be $\HH \triangleq \{ (v_1, v_2, \dots, v_m)~|~ m \in \mathbb{N} \cup \{0\},~ v_i \in \V\}$; when $m = 0$, we have an empty sequence. For two sequences $S_1 = (v_1, v_2, \dots, v_{m_1})$ and $S_2 = (u_1, u_2, \dots, u_{m_2})$ in $\HH$, we define a concatenation of $S_1$ and $S_2$ as
    \begin{equation*}
        S_1 \oplus S_2  \triangleq (v_1, v_2, \dots, v_{m_1}, u_1, u_2, \dots, u_{m_2}).
    \end{equation*}
We say that $S_1 \preceq S_2$ if $S_2$ can be rewritten as $S_1 \oplus S_3$ for some $S_3 \in \HH$. For sequence $S$, we use $\V(S)$ to denote the set of nodes in sequence $S$. By slightly abusing the notation, we use $|S|$ to denote the number of elements in sequence $S$. In addition, we often use $v$ to denote a singleton sequence $(v)$ when there is no confusion
    
In the following, unless stated otherwise, we only consider sequences with unique nodes. For sequence $S$, let $y_i(S)$ denote the total flow assigned to the $i$-th node and let $\boldsymbol{y}(S)$ denote a given feasible resource allocation vector for sequence $S$. By slightly abusing the notation, we use $\boldsymbol{X}(S)$ to denote the flow assignment matrix of sequence $S$. Consider any node $v_j$ in sequence $S$, with $j = 1, \dots, |S|$. Given a fixed resource allocation $\boldsymbol{y}(S)$, we define a fractional resource allocation of node $v_j$ as the solution of the following problem:
\begin{equation}
    \label{eq:nodeCapacityAllocation}
    \begin{aligned}
    & \underset{\boldsymbol{X}(S)}{\text{maximize}} \quad  \sum_{f\in \F}  x_f^{v_j}\\
    & \text{subject to} \\
    & \sum_{v \in \V_f \cap \V(S)} x_f^v \leq \lambda_f, && \forall f \in \F \text{ and } v \in \V_f, \\
    &  x_f^v=0, && \forall f \in \F \text{ and } v \notin \V_f, \\
    &  \sum_{\phi \in \Phi} \beta_\phi^r \sum_{f \in \F(\phi)} x_f^{v_i} \leq c_{v_i}^r, && \forall  r \in \R \text{ and }  i=1, \dots, |S|, \\
    &  \sum_{f \in \F} x_f^{v_i} =  y_i(S), && i = 1, \dots, |S| \text{ and } i \neq j.
    \end{aligned}
\end{equation}
    \noindent In Problem \eqref{eq:nodeCapacityAllocation}, we want to maximize the total traffic that can be assigned to node $v_j$ while satisfying the given resource allocation $\boldsymbol{y}(S)$ of all other nodes in $S$.
    
    Now, consider sequence $S$. The resource allocation of nodes in $S$ is presented in Algorithm \ref{alg:fractional}. Algorithm \ref{alg:fractional} starts by initializing the total traffic assigned to each node to zero (i.e., set $y_i(S) = 0,~ i = 1, \dots, |S|$), and then iterates over nodes in sequence $S$ according to their order. In iteration $i$, it computes the resource allocation of node $v_i$ by solving Problem \eqref{eq:nodeCapacityAllocation} given $\boldsymbol{y}(S)$, and then update $y_i(S)$ according to the obtained solution. We use $\hat{\boldsymbol{X}}(S)$ to denote the flow assignment matrix of sequence $S$ at the end of Algorithm \ref{alg:fractional}. Similarly, we use $\boldsymbol{\hat{y}}(S)$ to denote the resource allocation vector for sequence $S$ at the end of Algorithm \ref{alg:fractional}. For any sequence $S$, we define function $R_4(S)$ to be the total traffic assigned by Algorithm \ref{alg:fractional} for nodes in sequence $S$, i.e., 
    \begin{equation}
    \label{eq:R_4_S_nodes_based}
        R_4(S) \triangleq \sum_{i = 1}^{|S|} \hat{y}_i(S).
    \end{equation}
    Note that for sequences with repeated nodes, the value of $R_4(S)$ is the same as $R_4(\bar{S})$, where $\bar{S}$ is obtained from sequence $S$ by removing all the later appearances of the same node. Then, the relaxed version of Problem \eqref{eq:relaxedPlacement} becomes the following:
    \begin{equation}\tag{$Q4$}
    \label{eq:sequencePlacement}
    \begin{aligned}
    & \underset{S \in \HH}{\text{maximize}} \quad   {R_4}( S)\\
    & \text{subject to } |S| \leq k.
    \end{aligned}
    \end{equation}  

     \begin{algorithm}[t]
        \caption{Iterative resource allocation}
        \label{alg:fractional}
        \begin{algorithmic}[1]
            \Statex {Input: sequence of nodes $S$, set of flows $\F$, amount of resources $c_v^r$, flow rates $\lambda_f$, and flow demands $\beta_f^r$}
            \Statex Output: $\boldsymbol{\hat{y}}(S)$ and $\hat{\boldsymbol{X}}(S)$
            \State Initialize: $y_i(S) = 0,~ i = 1, \dots, |S|$
            \For {$i =1$ to $|S|$}
            \State $\boldsymbol{X}(S) \leftarrow$ solve Problem \eqref{eq:nodeCapacityAllocation} for node $v_i$ given $\boldsymbol{y}(S)$
            \State {Set $y_i(S) = \sum_{j = 1}^{F} x_{f_j}^{v_i}$}
            \EndFor
            \State {$\hat{\boldsymbol{X}}(S) = \boldsymbol{X}(S),~ \boldsymbol{\hat{y}}(S) = \boldsymbol{y}(S)$}
        \end{algorithmic}
    \end{algorithm}
   
    Next, we will show  that function $R_4(S)$ is a $1/2$-approximation of function $R_3(\U)$ for any sequence $S$ that is a permutation of nodes in set $\U$. This ensures that an optimal solution for Problem \eqref{eq:sequencePlacement} is a $1/2$-approximation solution of Problem \eqref{eq:relaxedPlacement}.  Moreover, in the next section, we will utilize the relaxed problem \eqref{eq:sequencePlacement} to design efficient algorithms for the multi-VPRA problem \eqref{eq:mainProblem}.
    
    First, we present Lemmas \ref{lemma:prefixequal} and \ref{lemma:v_from_lambda_S}, which will be used in the proof of the approximation ratio of Algorithm \ref{alg:fractional} and in establishing that function $R_4$ is sequence-submodular in the next section.  
\begin{lemma}
  \label{lemma:prefixequal}
  Consider $S_1$, $S_2$, and $S_3$ in $\HH$, such that $S_3 \preceq S_1$ and $S_3 \preceq S_2$. Applying Algorithm \ref{alg:fractional} to $S_1$ and $S_2$, respectively, yields $\hat{y}_i(S_1) = \hat{y}_i(S_2)$, for $i = 1, \dots, |S_3|$.
  \end{lemma}
  \begin{proof}
       See Section~\ref{proof:prefixequal}.
  \end{proof}
  
Before we present Lemma \ref{lemma:v_from_lambda_S}, we define some additional notations. By slightly abusing the notation, we use $\hat{\boldsymbol{X}}(\U)$ to denote an optimal flow assignment matrix of nodes $\U$ after solving Problem \eqref{eq:relaxedAllocation}. We use $\hat{x}_j(\U)$ to denote the total traffic assigned from flow $f_j$ to nodes $\U$, i.e.,
\begin{equation}
\label{eq:gamma_j_U}
    \hat{x}_j(\U) = \sum_{v \in \U} \hat{x}_{f_j}^v(\U),
\end{equation}
and we define $\hat{\boldsymbol{x}}(\U) \triangleq [\hat{x}_1(\U), \dots, \hat{x}_F(\U)]$. We can express $R_3(\U)$ in terms of $\hat{\boldsymbol{x}}(\U)$ as follows:
   \begin{equation}
       \label{eq:R_3_U_flows_based}
       R_3(\U) = \sum_{j =1}^{F} \hat{x}_j(\U).
   \end{equation}
Similarly, given $\hat{\boldsymbol{X}}(S)$, which is the flow assignment matrix in sequence $S$ at the end of Algorithm \ref{alg:fractional}, we use $\hat{x}_j(S)$ to denote the total traffic assigned from flow $f_j$ to nodes of sequence $S$, i.e.,
\begin{equation}
\label{eq:gamma_j_S}
    \hat{x}_j(S) = \sum_{i = 1}^{|S|} \hat{x}_{f_j}^{v_i}(S),
\end{equation}
and $\hat{\boldsymbol{x}}(S) \triangleq [\hat{x}_1(S), \dots, \hat{x}_F(S)]$. We can express $R_4(S)$ in terms of $\hat{\boldsymbol{x}}(S)$ as follows:
   \begin{equation}
       \label{eq:R_4_S_flows_based}
       R_4(S) = \sum_{j =1}^{F} \hat{x}_j(S).
   \end{equation}
   
We use $R_3(\U | \boldsymbol{x})$ to denote the value of function $R_3(\U)$ from a given flow rate vector $\boldsymbol{x}$; We use $R_4(S | \boldsymbol{x})$ in a similar manner. Note that $R_3(\U)$ is equivalent to $R_3(\U | \boldsymbol{\lambda})$; similarly, we have $R_4(S) = R_4(S | \boldsymbol{\lambda})$. In the sequel, we consider element-wise operations on vectors. For two flow rate vectors $\boldsymbol{x_1}$ and $\boldsymbol{x_2}$ such that $\boldsymbol{x_1} \leq \boldsymbol{x_2}$, we have
\begin{equation}
    \label{eq:compare_two_vectors_R_3}
    R_3(\U | \boldsymbol{x_1}) \leq R_3(\U | \boldsymbol{x_2}),
\end{equation}
where the inequality holds because any feasible solution to Problem \eqref{eq:relaxedAllocation} for node $\U$ given $\boldsymbol{x_1}$ is also a feasible solution to Problem \eqref{eq:relaxedAllocation} for node $\U$ given $\boldsymbol{x_2}$. In addition, for any node $u \in \V$, we have $R_4(u) = R_3(\{u\})$ because when applying Algorithm \ref{alg:fractional} to a singleton sequence, the equality constraints of Problem \eqref{eq:nodeCapacityAllocation} are irrelevant, which makes Problem \eqref{eq:nodeCapacityAllocation} equivalent to Problem \eqref{eq:relaxedAllocation}. As a result, we obtain the following:
\begin{equation}
    \label{eq:compare_two_vectors_R_4}
    R_4(u | \boldsymbol{x_1}) \leq R_4(u | \boldsymbol{x_2}).
\end{equation}
Note that Eq. \eqref{eq:compare_two_vectors_R_4} does not hold for non-singleton sequence in general.

Next, we present Lemma \ref{lemma:v_from_lambda_S}.
  \begin{lemma}
  \label{lemma:v_from_lambda_S}
      Consider $S_1$ and $S_2$ in $\HH$ such that $S_1 \preceq S_2$. For any node $u \notin \V(S_1)$, we have
      \begin{equation}
      \label{eq:v_from_lambda_S}
          R_4(u| (\boldsymbol{\lambda} - \hat{\boldsymbol{x}}(S_2))) \leq R_4(S_1 \oplus u| \boldsymbol{\lambda}) - R_4(S_1 | \boldsymbol{\lambda}).
      \end{equation}
  \end{lemma}
   \begin{proof}
    See Section~\ref{proof:v_from_lambda_S}.
   \end{proof}
   
The approximation ratio of Algorithm \ref{alg:fractional} is stated in the following lemma.  
\begin{lemma}
   \label{lemma:setToSequence}
     For a given set of nodes $\U$, let $\mathcal{P} (\U)$ be the set of all permutations of nodes $\U$. For any $S \in \mathcal{P} (\U)$, we have $\frac{1}{2} R_3(\U) \leq R_4(S) \leq R_3(\U)$.
   \end{lemma}
   \begin{proof}
    See Section~\ref{proof:setToSequence}.
   \end{proof}

   \section{Proposed Algorithms}
    \label{sec:MultiVPRAALgorithm}
    In this section, we design two algorithms that approximately solve the multi-VPRA problem \eqref{eq:mainProblem}. The main idea is to apply the two-level relaxation introduced in the previous section on the original non-relaxed problem \eqref{eq:mainProblem}. By doing so, we can show that the objective function of the relaxed placement subproblem \eqref{eq:sequencePlacement} is forward-monotone, approximate backward-monotone, and sequence-submodular (to be defined in Section \ref{subsec:sequence_submodular}).
    In this case, the relaxed placement subproblem can be approximately solved using an efficient greedy algorithm. Moreover, the relaxed allocation subproblem becomes a Linear Program (LP), which can also be solved efficiently in polynomial time. However, the solution to the relaxed problem is for the case where any fraction of the processed flows is counted. In order to obtain a solution for the original multi-VPRA problem \eqref{eq:mainProblem}, where only the fully processed flows are counted, we propose two approximation algorithms based on the primal-dual technique. 
    
    We use SSG-PRA and SSG-NRA to denote the algorithms we develop by combining the Sequence Submodular Greedy placement with the Primal-dual-based Resource Allocation and the Node-based Resource Allocation, respectively. We describe the algorithms in a unified framework presented in Algorithm \ref{alg:proposed}. The difference is in the resource allocation subproblem (Line \ref{line:capacityAllocation}), where SSG-PRA algorithm uses a Primal-dual-based Resource Allocation (PRA) algorithm presented in Algorithm \ref{alg:PRA}, while SSG-NRA algorithm uses a Node-based Resource Allocation (NRA) algorithm presented in Algorithm \ref{alg:NRA}. We show that the SSG-PRA and SSG-NRA algorithms achieve an approximation ratio of $\frac{(e-1)(Z-1) } {4 e^2 Z (kR)^{1/(Z-1)}}$ and $\frac{(e-1)(Z-1)} {4e(Z-1+Z R^{1/(Z-1)})}$, respectively, where $Z$ (to be defined in Section \ref{sec:resource_allocation}) is the amount of resource compared to flow demand. 
    
    \begin{algorithm}[t]
        \caption{The SSG-PRA  and SSG-NRA algorithms}
        \label{alg:proposed}
        \begin{algorithmic}[1]
            \Statex {Input: set of nodes $\V$, set of flows $\F$, amount of resources, flow rates, flows demand, and budget $B$}
            \Statex Output: set of VNF-nodes $\U$ and resource allocation $\boldsymbol{X}(\U)$
            \State \textbf{Relaxed Problem}: relax function $J_1(\U, \boldsymbol{X}(\U))$ to become $R_1(\U, \boldsymbol{X}(\U))$ (first-level relaxation), and relax function $R_3(\U)$ to function $R_4(S)$ (second-level relaxation)
            \State \textbf{Placement Subproblem}: solve Problem \eqref{eq:sequencePlacement} using the sequence submodular greedy algorithm (Algorithm \ref{alg:SSG}) to obtain $S$; let $\U = \V(S)$
            \State \textbf{Resource Allocation}: use either the PRA algorithm (Algorithm \ref{alg:PRA}) or the NRA algorithm (Algorithm \ref{alg:NRA}) to obtain resource allocation $\boldsymbol{X}(\U)$ for nodes $\U$ \label{line:capacityAllocation}
        \end{algorithmic}
    \end{algorithm}

   \subsection{Preliminary Results}
   \label{subsec:sequence_submodular}
   
   \begin{algorithm}[t]
        \caption{Sequence Submodular Greedy (SSG) algorithm}
        \label{alg:SSG}
        \begin{algorithmic}[1]
            \Statex {Input: nodes $\V$, $k$}
            \Statex {Initialization: $S = ()$}
            \Statex Output: $S$
            \While {$|S| < k$}
            \State {$S = S \oplus \argmax_{v \in \V} (h(S \oplus v) - h(S))$}
            \EndWhile
        \end{algorithmic}
    \end{algorithm}
     In this subsection, we present results related to sequence submodular functions, which will be used to derive a placement algorithm for Problem \eqref{eq:sequencePlacement}. Note that the definitions and results presented in this subsection generalize to sequences with repeated nodes. We start with some definitions. A function from sequences to real numbers, $h : \HH \rightarrow \mathbb{R}$, is sequence-submodular if
        \begin{equation}
        \label{eq:sequence_submodular}
        \begin{aligned}
            &\forall S_1, S_2 \in \HH, \text{ such that }  S_1 \preceq S_2, ~ \forall v \in \V,  \\
            & h(S_1 \oplus v) - h(S_1) \geq h(S_2 \oplus v) - h(S_2).
            \end{aligned}
        \end{equation}
    Also, function $h$ is forward-monotone if
        \begin{equation}
            \label{eq:forward_monotone}
            \forall S_1, S_2 \in \HH, ~ h(S_1 \oplus S_2) \geq h(S_1),
        \end{equation}    
   and function $h$ is backward-monotone if
        \begin{equation}
            \label{eq:backward_monotone}
            \forall S_1, S_2 \in \HH, ~ h(S_1 \oplus S_2) \geq h(S_2).
        \end{equation}

   Consider the problem of selecting a sequence $S$ of length $k$ that maximizes function $h(S)$, i.e., 
   \begin{equation}\tag{$W$}
   \label{eq:sequence_selection}
        \max_{S \in \HH: |S| \leq k} h(S).
   \end{equation}
    Although Problem \eqref{eq:sequence_selection} is NP-hard, it has been shown in \cite{streeter2009online} that for function $h$ that is forward-monotone, backward-monotone, and sequence-submodular, the \emph{Sequence Submodular Greedy} (SSG) algorithm, presented in Algorithm \ref{alg:SSG}, achieves an approximation of $(1 - 1/e)$. Algorithm \ref{alg:SSG} starts with an empty sequence $S$ and greedily adds a node that has the largest incremental value to sequence $S$ until $|S| = k$. 
    
    However, some functions may only satisfy an approximate version of the backward-monotone property (e.g., function $R_4(S)$ as shown in Lemma \ref{lemma:submodularity}). Therefore, we generalize the backward-monotone property as follows: For $\alpha \in (0, 1]$, function $h$ is $\alpha$-backward-monotone if
        \begin{equation}
            \label{eq:alpha_backward_monotone}
            \forall S_1, S_2 \in \HH, ~ h(S_1 \oplus S_2) \geq \alpha h(S_2).
        \end{equation}
    Then, in the following theorem, we derive the approximation ratio of Algorithm \ref{alg:SSG} for function $h$ that satisfies the $\alpha$-backward-monotone property. This result will be used later to design an efficient algorithm for our VNF-node placement problem. In the sequel, for any Problem $(P)$, we use $\text{OPT}(P)$ to denote its optimal value.
    Without loss of generality, we assume that the value of an empty sequence is zero. 
   \begin{theorem}
   \label{theorem:SSG}
    Suppose that sequence function $h$ is forward-monotone, $\alpha$-backward-monotone, and sequence-submodular. Then, Algorithm \ref{alg:SSG} achieves an approximation ratio of $\alpha (1-1/e)$ for Problem \eqref{eq:sequence_selection}, i.e., $h(S) \geq \alpha(1-1/e) \text{OPT}\eqref{eq:sequence_selection}$.
   \end{theorem}
    \begin{proof}
    The proof follows a similar line of analysis as in \cite{alaei2010maximizing}. See Section~\ref{proof:SSG}.
    \end{proof}
    
   \subsection{Placement Algorithm}
    \label{subsec:placement}
      In this subsection, we prove that function $R_4(S)$ is forward-monotone, $\frac{1}{2}$-backward-monotone, and sequence-submodular. Then, using the property of sequence submodularity, we employ the SSG  algorithm (Algorithm \ref{alg:SSG}) for solving the placement subproblem. We start with the following lemma.
     
    \begin{lemma}
        \label{lemma:submodularity}
        The function $R_4(S)$ is forward-monotone, $\frac{1}{2}$-backward-monotone, and sequence-submodular.
    \end{lemma}
    \begin{proof}
    See Section~\ref{proof:submodularity}.
    \end{proof}
    
    We would like to point out that function $R_4$ is $\frac{1}{2}$-backward-monotone and that the bound of  
    $\frac{1}{2}$ is tight. We prove it through constructing a  problem instance in Section~\ref{example:lower_bound}. 

    Because of this useful property of sequence submodularity, Problem \eqref{eq:sequencePlacement} can be approximately solved using the SSG algorithm (Algorithm \ref{alg:SSG}). In the SSG algorithm, we start with an empty solution of VNF-nodes in $S$; in each iteration, we add a node that has the maximum marginal contribution to $S$, i.e., a node that leads to the largest increase in the value of the objective function $R_4(S)$. We repeat the above procedure until $k$ VNF-nodes have been selected. Note that if a node has been selected in a previous iteration, then its marginal contribution in any subsequent iteration will be zero. If at any iteration the marginal contribution of all nodes is zero, then we select a node that has not been selected before. In this way, we ensure that the selected sequence $S$ has no repeated nodes. To solve Problem \eqref{eq:relaxedPlacement}, we need a set of nodes rather than a sequence, but in order to use Algorithm \ref{alg:SSG} and take advantage of its approximation ratio, we select a sequence of unique nodes, which can be converted to a set. We state the performance of the SSG algorithm for Problem \eqref{eq:sequencePlacement} in the following theorem.

     \begin{theorem}
        \label{theorem:multiVPRAPlacement}
       The SSG algorithm achieves an approximation ratio of $\frac{1}{2}(1-1/e)$ for Problem \eqref{eq:sequencePlacement}, i.e., $R_4(S) \geq \frac{1}{2}(1-1/e)\text{OPT}\eqref{eq:sequencePlacement}$.
    \end{theorem}
    \begin{proof}
        \blue{The result follows immediately from Theorem~\ref{theorem:SSG} and Lemma~\ref{lemma:submodularity} by plugging $\alpha = \frac{1}{2}$ into Theorem~\ref{theorem:SSG}.}    
    \end{proof}

    \subsection{Resource Allocation Algorithms}
    \label{sec:resource_allocation}
    While solving the placement subproblem \eqref{eq:sequencePlacement}, the resource allocation is achieved by using Algorithm \ref{alg:fractional}, which allows partially processed flows to be counted. However, Problem \eqref{eq:mainProblem} requires flows to be fully processed. Therefore, we present two resource allocation algorithms that modify the resource allocation of the selected VNF-nodes while guaranteeing certain approximation ratios. Both algorithms are based on the primal-dual technique \cite{briest2011approximation}. We describe each of the algorithms in the following.
    
    For the selected sequence $S$ of VNF-nodes, let $\U = \V(S)$. We first provide a formulation of the optimal fractional resource allocation of VNF-nodes $\U$, which allows partially processed flows. Based on the dual of this formulation, we will present the two resource allocation algorithms. We define $\delta_f^r \triangleq \sum_{\phi \in  \Phi_f} \beta_\phi^r$ to be the total amount of resource $r$ needed to process a unit of flow $f$ by a set of network functions $\Phi_f$. We define the maximum demand across all flows as $d_{\text{max}} \triangleq \max_{f \in \F, r\in \R} \delta_f^r \lambda_f$. Then, for each flow $f$ we define the normalized total demand of resource $r$ as $d_f^r \triangleq \delta_f^r \lambda_f /d_{\text{max}}$. In addition,  for each VNF-node $v$, we define the normalized total amount of resource $r$ as $\bar{c}_v^r \triangleq  c_v^r/d_{\text{max}}$. Finally, we define $Z \triangleq \min_{v \in S, r\in \R} \bar{c}_v^r$ as a measure of the available resource compared to flow demand (we call it resource stretch). We use $a_f^v$ to denote the fraction of flow $f$ that is assigned to VNF-node $v$. The optimal fractional resource allocation of VNF-nodes $\U$ can be formulated as:
    \begin{equation}
    \begin{aligned}
            &\max_{a_f^v} \sum_{f \in \F} \lambda_f \sum_{v \in \U \cap \V_f} a_f^v\\
            & \text{subject to} \\
            & \sum_{f \in \F} d_f^r a_f^v \leq \bar{c}_v^r, ~ &&\forall r \in \R ~ \text{ and } v \in \U, \\
            & \sum_{v \in S_f} a_f^v \leq 1, ~ &&\forall f\in \F, \\
            & a_f^v \geq 0, ~ &&\forall f \in \F \text{ and } v \in \U.
    \end{aligned}
         \label{eq:primal}
    \end{equation}
        \noindent The corresponding dual linear program is
        \begin{equation}
    \begin{aligned}
            &\min_{b_v^r, z_f} \sum_{v \in \U} \sum_{r \in \R} \bar{c}_v^r b_v^r + \sum_{f \in \U \cap \V_f} z_f\\
            & \text{subject to} \\
            & z_f + \sum_{r \in \R} d_f^r  b_v^r \geq \lambda_f, ~ \forall f \in \F ~ \text{ and } v \in \U \cap \V_f, \\
            & b_v^r, z_f \geq 0, ~~~~~~~~~~~~~~\forall v \in \V,~ r \in \R ~ \text{ and } f \in \F.
    \end{aligned}
     \label{eq:dual}
    \end{equation}
    
    \subsubsection{Primal-dual-based Resource Allocation (PRA)}
     For the VNF-nodes $\U$ that are selected by the SSG algorithm, we modify their resource allocation to guarantee fully processed flows.  We propose a primal-dual-based resource allocation algorithm, which is adapted from a multi-commodity routing algorithm proposed in \cite{briest2011approximation} and based on the dual formulation \eqref{eq:dual}. The main idea is to view the dual variable $b_v^r$ as a price of resource $r$ at VNF-node $v$. The algorithm chooses a VNF-node $v_f$ with the minimum total cost for each flow. Then, it picks a flow that maximizes the relative value of $\lambda_f$ compared to the weighted cost and assigns that flow to the associated node. Then the price of each resource of the selected VNF-node is updated accordingly. The update of the price $b_v^r$ is designed in a way such that if the limited resource is violated, then the stopping condition is satisfied from the previous iteration. The algorithm stops when all flows are assigned or when $\sum_{v \in \U} \sum_{r \in \R} \bar{c}_v^r b_v^r \geq e^{Z-1}R|\U|$.  The update of price $b_v^r$ is also implemented in a way such that it maintains the value of the dual problem within a range of the value of the primal problem. Then, by weak duality, this establishes the approximation ratio of the primal-dual algorithm.  

We use $\pi_{\text{PRA}}^{\U}$ to denote the total traffic assigned to VNF-nodes $\U$ by the PRA algorithm. The approximation ratio of the PRA algorithm with respect to function $R_4(S)$ is stated in the following Lemma.
\begin{lemma}
    \label{lemma:PRAapproximation}
    The approximation ratio of the  PRA  algorithm is  $\pi_\text{PRA}^{\U} \geq \frac{Z-1}{e Z (kR)^{1/(Z-1)}} R_4(S)$.
\end{lemma}
\begin{proof}
    Recall that we use $\text{OPT}\eqref{eq:primal}$ to denote the optimal value of Problem \eqref{eq:primal}. The proof follows from the following:
    \begin{equation}
            \begin{aligned}
            \pi_\text{PRA}^{\U} &\stackrel{\text{(a)}}{\geq} \frac{Z-1}{e Z (|\U| R)^{1/(Z-1)}} \text{OPT}\eqref{eq:primal}\\
            &\stackrel{\text{(b)}}{\geq} \frac{Z-1}{e Z (k R)^{1/(Z-1)}} \text{OPT}\eqref{eq:primal} \\
            &\stackrel{\text{(c)}}{\geq} \frac{Z-1}{e Z (k R)^{1/(Z-1)}} R_4(S).
            \end{aligned}
    \end{equation} 
    The primal-dual algorithm has been shown to achieve the approximation ratio in (a) with respect to any fractional solution \cite[Lemma 5.7, Theorem 5.1]{briest2011approximation}; (b) follows because $k \geq |\U|$; (c) follows from the fact that the value $R_4(S)$ is upper bounded by $\text{OPT}\eqref{eq:primal}$.
\end{proof}

When $Z$ goes to infinity, then the algorithm has an approximation ratio of $1/e$. The time complexity of the PRA algorithm is $O(|\U| F^2)$.

     \begin{algorithm}[t]
        \caption{Primal-dual-based resource allocation (PRA)}
        \label{alg:PRA}
        \begin{algorithmic}[1]
            \State {Input: VNF-nodes $\U$, set of flows $\F$, normalized resources $\bar{c}_v^r$, flow rates $\lambda_f$, normalized flow demands $d_f^r$.}
            \State Initialization: $b_v^r = 1/\bar{c}_v^r, ~ \forall r \in \R, v \in \U,~ x_f^v =0,  \forall f \in \F \text{ and } v \in \U$ 
            \State Output: $\boldsymbol{X}(\U)$
           \Repeat
           \For {$f \in \F$} 
           \State $v_f = \argmin_{v \in \U \cap \V_f} \{\sum_{r \in R} b_v^r\};$
           \EndFor
           \vspace{1mm}
           \State $f^\prime = \argmax_{f \in \F}\{ \frac{\lambda_f}{\sum_{r \in \R} d_f^r b_{v_f}^r}\};$
           \vspace{1.5mm}
           \State $x_{f^\prime}^{v_{f^\prime}} = \lambda_{f^\prime};$
           \vspace{1.5mm}
           \State $\F = \F \setminus \{f^\prime\};$
           \vspace{1.5mm}
           \State update $b_{v_{f^\prime}}^r = b_{v_{f^\prime}}^r(e^{Z-1}R|\U|)^{d_{f^\prime}^r / (\bar{c}_{v_{f^\prime}}^r -1 )}, ~ \forall r \in \R;$
           \vspace{1.5mm}
           \Until {$\sum_{v \in \U} \sum_{r \in \R} \bar{c}_v^r b_v^r \geq e^{Z-1} R|\U|$ or $\F = \emptyset$;}
        \end{algorithmic}
    \end{algorithm}
    
\subsubsection{Node-based Resource Allocation (NRA)}
The approximation ratio of the PRA algorithm depends on two parameters: the budget $k$  and the resource stretch $Z$. If $k$ is large and $Z$ is small, then the approximation ratio of the PRA algorithm becomes small. However, if $Z$ is large enough, then it will offset the effect of large $k$. Therefore, we design another algorithm, node-based resource allocation algorithm (NRA), which removes the dependence on $k$ but adds a constant factor to the approximation ratio. The main idea of the NRA algorithm is to make the resource allocation of each VNF-node separately based on any order. For each VNF-node in $\U$, its resources are allocated using the primal-dual technique by considering the remaining unassigned flows. The detail of the NRA algorithm is presented in Algorithm \ref{alg:NRA}. Similar to the PRA algorithm, we view the dual variable $b_v^r$ as a price for each resource. The difference here is that we consider each VNF-node separately and try to assign flows with the largest ratio of the rate $\lambda_f$ compared to the weighted demand $\sum_{r \in \R} d_f^r b_v^r$. 
\begin{algorithm}[t]
        \caption{Node-based Resource Allocation (NRA)}
        \label{alg:NRA}
        \begin{algorithmic}[1]
            \State {Input: VNF-nodes $\U$, set of flows $\F$, normalized resources $\bar{c}_v^r$, flow rates $\lambda_f$, normalized flow demands $d_f^r$.}
            \State Initialization: $x_f^v = 0, \forall f \in \F, v \in \U$
            \State Output: $\boldsymbol{X}(\U)$
            \For {each VNF-node $v \in \U$}
                \State Initialization: $b_v^r = 1/\bar{c}_v^r, ~ \forall r \in \R$
               \Repeat
               \State $f^\prime = \argmax_{f \in \F}\{ \frac{\lambda_f}{\sum_{r \in \R} d_f^r b_v^r}\};$
               \vspace{1.5mm}
               \State $x_{f^\prime}^v = \lambda_{f^\prime};$
               \vspace{1.5mm}
               \State $\F = \F \setminus \{f^\prime\};$
               \vspace{1.5mm}
               \State update $b_v^r = b_v^r(e^{Z-1}R)^{d_f^r / (\bar{c}_v^r -1 )}, ~ \forall r \in \R;$
               \vspace{1.5mm}
               \Until {$\sum_{r \in \R} \bar{c}_v^r b_v^r \geq e^{Z-1}R$ or $\F = \emptyset$;}
            \EndFor
        \end{algorithmic}
    \end{algorithm}
    
We use $\pi_{\text{NRA}}^{\{v\}}$ to denote the total traffic assigned to VNF-node $v$ by the NRA algorithm and define $\pi_{\text{NRA}}^\U \triangleq \sum_{v \in \U} \pi_{\text{NRA}}^{\{v\}}$. We state the approximation ratio of the NRA algorithm in the following lemma.
\begin{lemma}
    The approximation ratio of the  NRA  algorithm is $\pi_\text{NRA}^\U \geq \frac{Z-1}{Z-1+e Z R^{1/(Z-1)}} R_4(S)$.
    \label{lemma:NRAapproximation}
\end{lemma}
\begin{proof}
     First, we define additional notations. Let $\F^\prime \subseteq \F$ denote the set of unassigned flows by the end of Algorithm \ref{alg:NRA} and $\F_v$ denote the set of unassigned flows right before considering VNF-node $v$ by Algorithm \ref{alg:NRA}. By slightly abusing the notation, we use $\text{OPT}(\{v\}|\bar{\F})$ to denote the optimal resource allocation of VNF-node $v$ considering only the subset of flows $\bar{\F}$. We have
            \begin{equation}
            \begin{aligned}
            R_4(S) & \stackrel{\text{(a)}} \leq  R_3(\U) \\
            & \stackrel{\text{(b)}} \leq \sum_{ v \in \U} \text{OPT}(\{v\}|\F) \\
            &\stackrel{\text{(c)}} = \sum_{f \in \F \setminus \F^\prime} \lambda_f  + \sum_{ v \in \U} \text{OPT}(\{v\}|\F^\prime)\\
           &\stackrel{\text{(d)}}{=}  \pi_\text{NRA}^\U  + \sum_{ v \in \U} \text{OPT}(\{v\}|\F^\prime)\\
          & \stackrel{\text{(e)}}{\leq} \pi_\text{NRA}^\U  + \sum_{ v \in \U} \text{OPT}(\{v\}|\F_v)\\
            &\stackrel{\text{(f)}}{\leq} \pi_\text{NRA}^\U + \sum_{ v \in \U} \frac{e Z}{Z-1}R^{1/(Z-1)} ~\pi_\text{NRA}^{\{v\}}\\
            & \leq\pi_\text{NRA}^\U + \frac{e Z}{Z-1}R^{1/(Z-1)} \pi_\text{NRA}^\U\\
            & = \frac{Z-1 +e Z R^{1/(Z-1)}}{Z-1} \pi_\text{NRA}^\U,
            \end{aligned}
            \end{equation} 
             where (a) follows from Lemma \ref{lemma:setToSequence}; (b) holds because we consider each node individually with all flows $\F$; (c) holds because we can consider what can be assigned from a subset of flows $\F^\prime$ and add to it all other flows $\F \setminus \F^\prime$; (d) holds because flows $\F \setminus \F^\prime$ are all assigned by the NRA algorithm;  (e) holds because $\F_v$ is a superset of $\F^\prime$. For (f), the NRA algorithm for a single VNF-node achieves an approximation ratio of $\frac{e Z}{Z-1}R^{1/(Z-1)}$ with respect to any fractional solution \cite[Lemma 5.7, Theorem 5.1]{briest2011approximation}, so (f) holds. 
        \end{proof}
        
When $Z$ goes to infinity, then the approximation ratio is $1/(e+1)$. The time complexity of the NRA algorithm is $O(F^2)$. 

 \subsection{Main Results}
    \label{subsec:mainResult}
    We state our main results in Theorems \ref{theorem:mainResult} and \ref{theorem:mainResult2}.
    \begin{theorem}
        \label{theorem:mainResult}
        The SSG-PRA algorithm has an approximation ratio of $\frac{(e-1)(Z-1)}{4 e^2 Z (k R)^{1/(Z-1)}}$ for Problem \eqref{eq:mainProblem} and becomes $\frac{e-1}{4e^2}$ when $Z \rightarrow \infty$. 
    \end{theorem}
        \begin{proof}
       The SSG-PRA algorithm has two main components: 1) VNF-nodes placement and 2) resource allocation. We use $OPT(P)$ to denote the optimal value of any problem $(P)$. We start with the result of the VNF-nodes placement using the SSG algorithm. For sequence $S$ that is selected by the SSG algorithm, we have the following result:
        \begin{equation}  
        \label{eq:placementresult2}      
        \begin{aligned}
        R_4(S) &  \stackrel{\text{(a)}}{\geq}  \frac{1}{2}(1-1/e) \text{OPT}(\text{\ref{eq:sequencePlacement}})\\
        &  \stackrel{\text{(b)}}{\geq}  \frac{1}{4}(1-1/e)\text{OPT}(\text{\ref{eq:relaxedPlacement}})\\
        & \stackrel{\text{(c)}}{=} \frac{1}{4}(1-1/e)\text{OPT}(\text{\ref{eq:relaxedProblem}}) \\
        & \stackrel{\text{(d)}}{\geq} \frac{1}{4}(1-1/e)\text{OPT}(\text{\ref{eq:mainProblem}}),
        \end{aligned}
        \end{equation}
        where (a) is due to Theorem \ref{theorem:multiVPRAPlacement}, (b) holds from Lemma \ref{lemma:setToSequence}, (c) holds because an optimal resource allocation is assumed for the objective function of Problem \eqref{eq:relaxedPlacement}, and (d) holds because Problem \eqref{eq:relaxedProblem} is a relaxed version of Problem \eqref{eq:mainProblem}. 
        
        The second component of the SSG-PRA algorithm is the resource allocation using the PRA algorithm for the sequence of VNF-nodes $S$ selected by the SSG. We have the following result:
        \begin{equation}
        \label{eq:mainTheorem}
        \begin{aligned}
        \pi_\text{PRA}^{\U} & \stackrel{\text{(a)}}{\geq} \frac{Z-1}{e Z (k R)^{1/(Z-1)}} R_4(S) \\
        & \stackrel{\text{(b)}}{\geq} \frac{(e-1)(Z-1)}{4 e^2 Z (k R)^{1/(Z-1)}} \text{OPT}(\text{\ref{eq:mainProblem}}),
        \end{aligned}
        \end{equation} 
        where (a) comes from the approximation ratio of the PRA algorithm in Lemma \ref{lemma:PRAapproximation}, and (b) holds from Eq. \eqref{eq:placementresult2}. 
        Therefore, the result of Theorem \ref{theorem:mainResult} follows.
        \end{proof}        
       
    \begin{theorem}
        \label{theorem:mainResult2}
        The SSG-NRA algorithm has an approximation ratio of $\frac{(e-1)(Z-1)}{4e(Z-1+e Z R^{1/(Z-1)})}$ for Problem \eqref{eq:mainProblem} and becomes $\frac{e-1}{4e^2+4e}$ when $Z \rightarrow \infty$. 
    \end{theorem}
    \begin{proof}
        The proof follows the same steps as the proof of Theorem \ref{theorem:mainResult}.
    \end{proof}
    
    \section{Numerical Results}
       \begin{figure*}[t]
        \centering
        \begin{subfigure}{0.32\linewidth}
            \centering
            \includegraphics[width=0.99\linewidth]{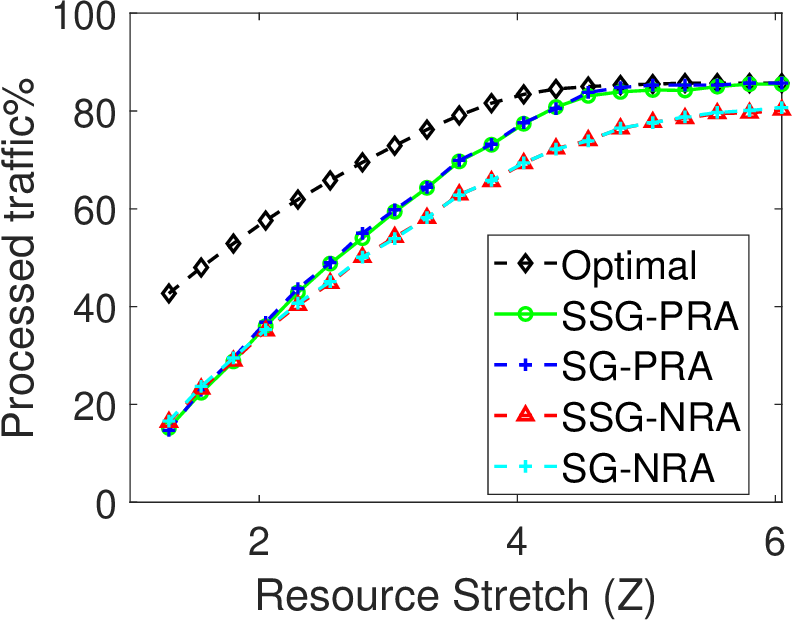}
            \caption{Budget = 3 VNF-nodes}
            \label{fig:Abilene_TotalFlow_3_132includeAll}
        \end{subfigure}
        \begin{subfigure}{0.32\linewidth}
            \centering
            \includegraphics[width=0.99\linewidth]{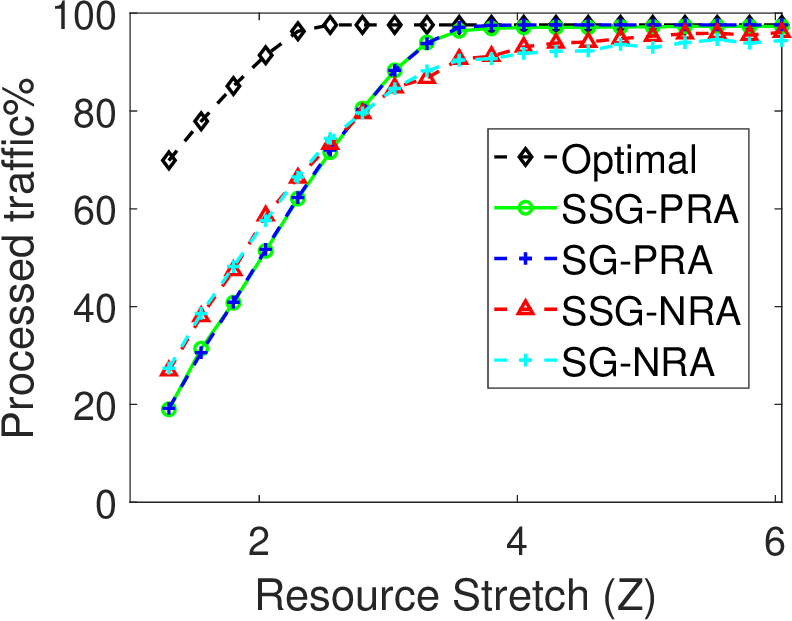}
            \caption{Budget = 6 VNF-nodes}
            \label{fig:Abilene_TotalFlow_6_132includeAll}
        \end{subfigure}
        \begin{subfigure}{0.32\linewidth}
            \centering
            \includegraphics[width=0.99\linewidth]{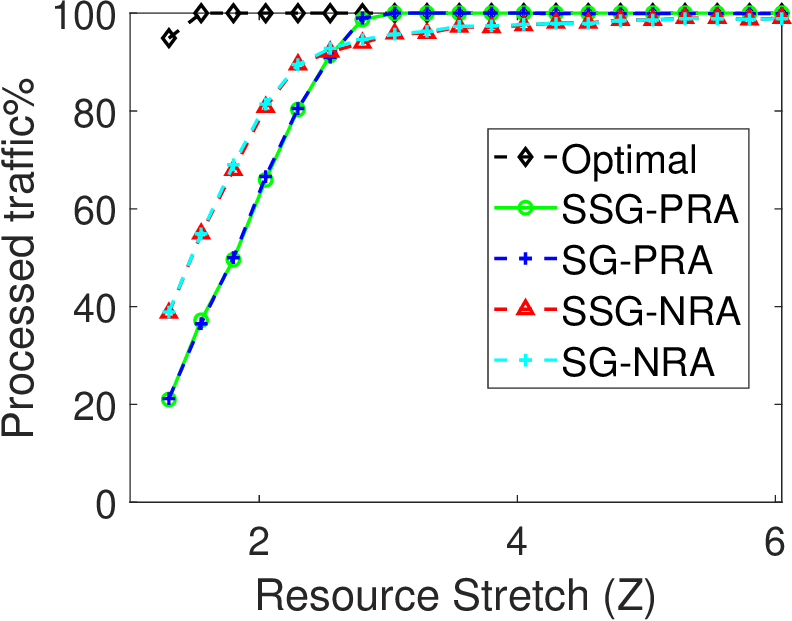}
            \caption{Budget = 10 VNF-nodes}
            \label{fig:Abilene_TotalFlow_10_132includeAll}
        \end{subfigure}
        \caption{Evaluation of Abilene dataset with different budget $k$ and resource stretch $Z$}
        \label{fig:abilene}
    \end{figure*} 
 
    In this section, we complement our theoretical analysis of the proposed algorithms with a trace-driven simulation study. We compare the proposed algorithms with the optimal solution, obtained by solving the Integer Linear Program (ILP) formulation \eqref{eq:mainProblem} using Gurobi solver (Gurobi 8.1.1). In addition, we conjecture that the objective function of placement subproblem \eqref{eq:relaxedPlacement} is submodular. Therefore, we present the following two heuristics (SG-PRA algorithm and SG-NRA algorithm) based on this conjecture. In both heuristics, the placement is implemented in a similar way to that of the SSG algorithm, called Submodular Greedy (SG) algorithm \cite{nemhauser1981maximizing}. Specifically, we start with an empty solution of VNF-nodes $\U$; in each iteration, we add a node that has the maximum marginal contribution to $\U$, i.e., a node that leads to the largest increase in the value of the objective function $R_3(\U)$. We repeat the above procedure until $k$ VNF-nodes have been selected. Then, the resource allocation is implemented using the PRA (resp., NRA) algorithm for the SG-PRA (resp., SG-NRA) algorithm. We evaluate all algorithms based on the percentage of the processed traffic achieved by them, which is defined as the ratio between the total volume of the traffic processed by the VNF-nodes and the total traffic volume. 
    
    Note that we present the results of the optimal solution found by an ILP solver as a benchmark for comparisons only. While the ILP solver seems to work reasonably well for some problem instances considered in our simulations, the multi-VPRA problem is NP-hard in general (Lemma \ref{lemma:nphardness}). That is, there is no guarantee that any problem instance can be efficiently solved, and it may take a prohibitively long time to solve the problem in the worst-case scenarios.
    
    
    \subsection{Evaluation Datasets}

    \subsubsection{Abilene Dataset}
     We consider the Abilene dataset collected from an educational backbone network in North America \cite{abilene}. The network consists of 12 nodes and 144 flows. Each flow rate was recorded every five minutes for 6 months. Also, OSPF weights were recorded, which allows us to compute the shortest path of each flow based on these weights. In our experiments, we set the flow rate to the recorded value of the first day at 8:00 pm. We consider two types of resources (i.e., $R =2$), and the demand of each flow is randomly chosen between 0 and 20 (i.e., $\delta_f^r \in [0, 20]$). The total available resource is set to the maximum total demand of flows $d_\text{max}$ multiplied by a scaling parameter $Z >1$. 
     
     \subsubsection{SNDlib Datasets}
    We also consider two other datasets from SNDlib \cite{orlowski2010sndlib}: Cost266 with 37 nodes and 1332 flows, and ta2 with 65 nodes and 1869 flows. For Cost266, the link's routing cost is available, so we use that to compute the shortest path of each flow. For ta2, we use hop-count-based shortest path. The setting of resources is the same as that of the Abilene dataset.
    
    \subsection{Evaluation Results}
      \begin{figure*}[t]
        \centering
        \begin{subfigure}[t]{0.24\linewidth}
            \centering
            \includegraphics[width=0.99\linewidth]{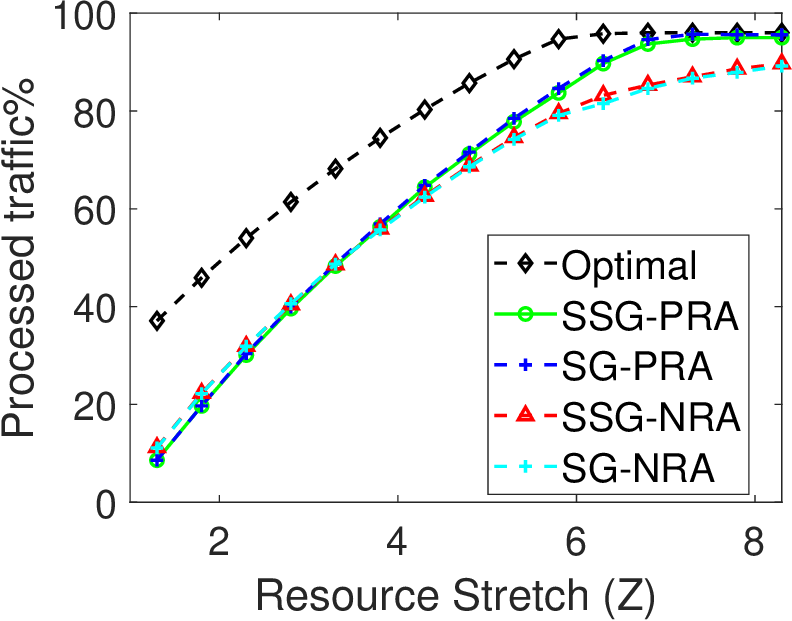}
            \caption{Budget=10 VNF-nodes \vspace{5mm}}
            \label{fig:cost266_TotalFlow_10_1332includeAll}
        \end{subfigure}
        \begin{subfigure}[t]{0.24\linewidth}
            \centering
            \includegraphics[width=0.99\linewidth]{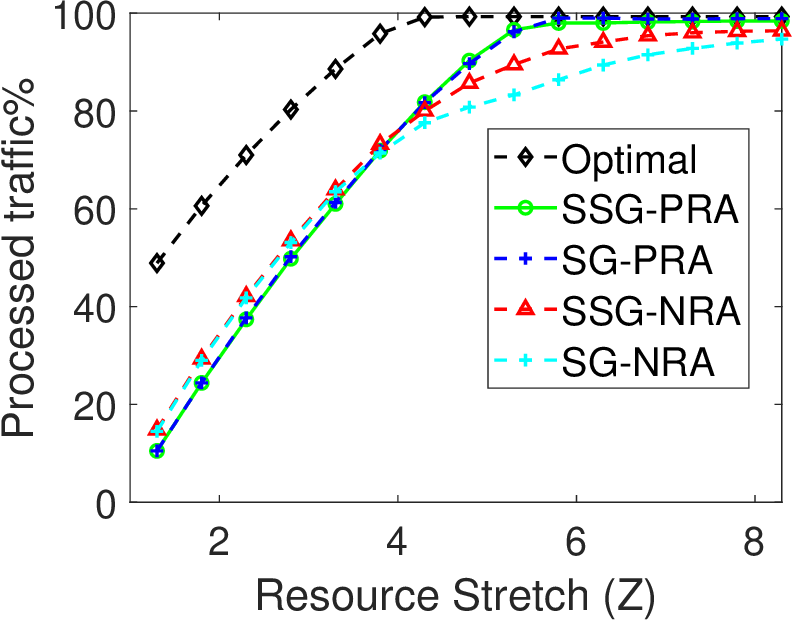}
            \caption{Budget = 15 VNF-nodes \vspace{5mm}}
            \label{fig:cost266_TotalFlow_15_1332includeAll}
        \end{subfigure}
        \begin{subfigure}[t]{0.24\linewidth}
            \centering
            \includegraphics[width=0.99\linewidth]{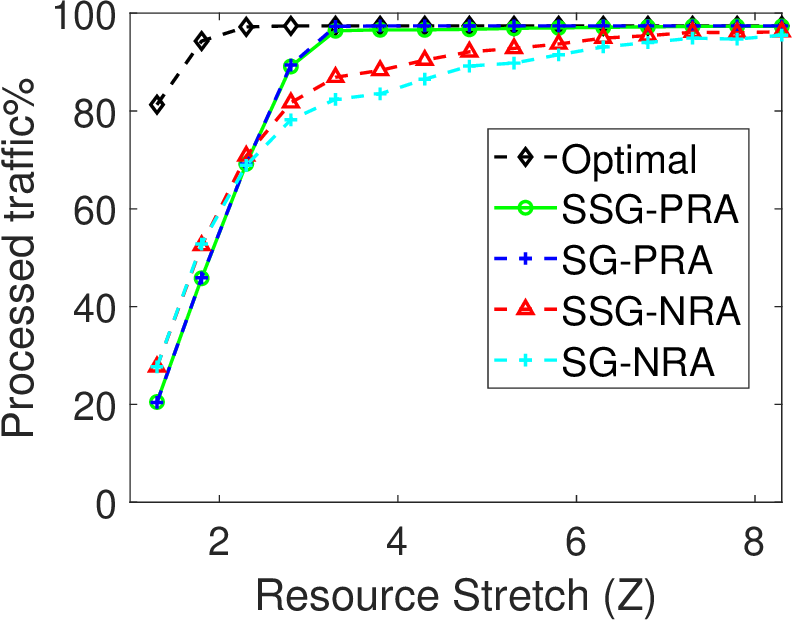}
            \caption{Budget = 10 VNF-nodes}
            \label{fig:ta2_TotalFlow_10_1869includeAll}
        \end{subfigure}
        \begin{subfigure}[t]{0.24\linewidth}
            \centering
            \includegraphics[width=0.99\linewidth]{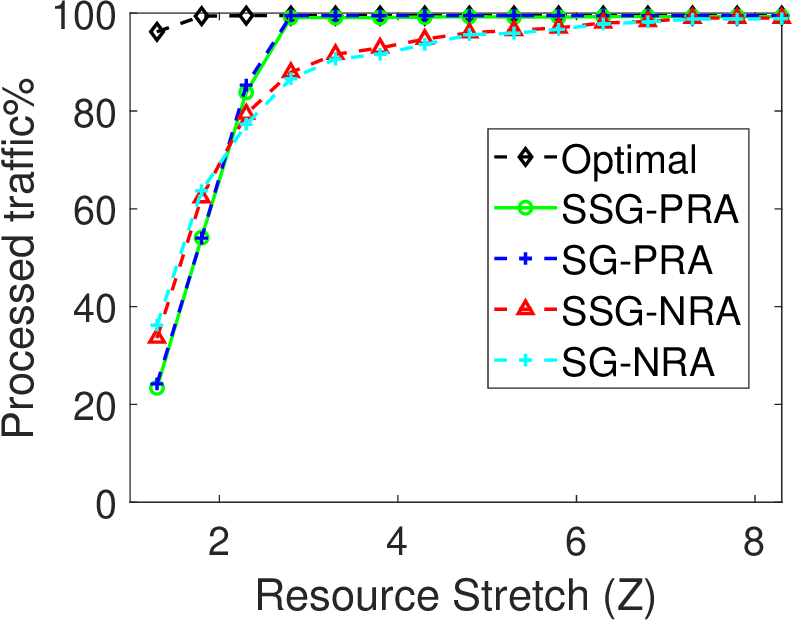}
            \caption{Budget = 15 VNF-nodes}
            \label{fig:ta2_TotalFlow_15_1869includeAll}
        \end{subfigure}
        \caption{Evaluation of Cost266 dataset (a-b) and ta2 dataset (c-d)}
        \label{fig:otherTopologies}
    \end{figure*}
    
    We start with the Abilene dataset, where we study the effect of having different values of resource stretch $Z$ and budget $B$. Remember that $Z$ is the ratio of the minimum available resource to the maximum flow demand. We consider a budget of  3, 6, and 10 VNF-nodes. The results are presented in Fig.~\ref{fig:abilene}. From the results, we make the following key observations.
    
    First, we can see that the simulation results for both the SSG-PRA and SSG-NRA algorithms agree with their approximation ratios presented in Theorems \ref{theorem:mainResult} and \ref{theorem:mainResult2} in that when the budget or $Z$ is small, the SSG-NRA performs better and vice versa. Specifically, we start with Fig.~\ref{fig:abilene}(\subref{fig:Abilene_TotalFlow_3_132includeAll}) when the budget is 3. When the amount of resources is small or there are flows with huge demand (i.e., $Z$ is small), the SSG-NRA algorithm is slightly better, but since the number of resources and nodes (i.e.,  $R|\U|$) is small anyway, it does not affect the performance of the SSG-PRA algorithm much. When $Z$ becomes large (either by having larger amount of resources or by having flows with smaller demand to make $Z \geq 4$), the effect of the terms $R|\U|^{Z-1}$ and $R^{Z-1}$ diminishes, but the effect of the constant term of the SSG-NRA algorithm remains, which corresponds to a slightly worse performance for larger $Z$. By doubling the budget to 6 VNF-nodes, we can see in Fig.~\ref{fig:abilene}(\subref{fig:Abilene_TotalFlow_6_132includeAll}) that the performance of the SSG-NRA algorithm is better than the SSG-PRA algorithm when $Z$ is small (i.e., $Z \leq 2.5$). This is because when $Z$ is small and $R |\U|$ is large, there is a high chance that the stopping condition of the PRA algorithm is satisfied early although some nodes still have large unused resources. In contrast, for the NRA algorithm, we consider nodes one by one, and if the stopping condition is satisfied early, it will only affect the node under consideration and the algorithm will continue allocating the resources of the other nodes. The same trend can also be seen in Fig.~\ref{fig:abilene}(\subref{fig:Abilene_TotalFlow_10_132includeAll}). 
    
    Second, although the SSG-NRA algorithm works better when $Z$ is small, sometimes it fails to reach the performance of the optimal solution even when $Z$ is large (see Fig.~\ref{fig:abilene}(\subref{fig:Abilene_TotalFlow_3_132includeAll})). Increasing the budget helps alleviating this problem with SSG-NRA algorithm, but still it needs at least twice the resource stretch $Z$ needed by the SSG-PRA algorithm to reach a similar performance of the optimal solution (see Figs.~\ref{fig:abilene}(\subref{fig:Abilene_TotalFlow_6_132includeAll}) and \ref{fig:abilene}(\subref{fig:Abilene_TotalFlow_10_132includeAll})). The proposed algorithms achieve at least $1/2$ of the optimal solution, which verifies our theoretical results. 
    
    Third, comparing the proposed algorithms with the two heuristics, we can see that the proposed algorithms perform almost the same as the heuristics. The proposed algorithms even work better in multiple occasions as for the SSG-NRA algorithm. That means even if our conjecture that $R_3(\U)$ is submodular is correct,  the loss by considering the second-level relaxation (i.e., the $1/2$-approximation factor in Theorem~\ref{theorem:multiVPRAPlacement}) is negligible. However, the second-level relaxation is important as it allows to draw a connection to the sequence submodular theory and establish the performance guarantee of the SSG algorithm. 

    Fourth, The results suggest that in order to gain the best performance in term of total processed traffic, ISPs have two options: 1) either to scale resources vertically by provisioning more resources at each node (i.e., makes $Z$ large); or 2) scale horizontally by deploying more VNF-nodes. Both of these options have shown promising performance as can be seen in Fig.~\ref{fig:abilene}.
    
    Furthermore, we extend the evaluation to other datasets with a larger number of nodes and flows in Fig.~\ref{fig:otherTopologies}. We consider Cost266 dataset (37 nodes and 1332 flows) and ta2 dataset (65 nodes and 1869 flows). We consider two settings of budget of 10 and 15 VNF-nodes. Comparing with the proposed algorithms, we can see a similar trend to that of Fig.~\ref{fig:abilene} in that the SSG-NRA algorithm works better for a smaller $Z$ and vice-versa for the SSG-PRA algorithm. Comparing both algorithms with the optimal solution, the proposed algorithms are also within $1/2$ of the value achieved by the optimal solution. In addition, we can see that the heuristics and the proposed algorithms perform very similarly to each other and that no algorithm constantly dominates the other. We note that although the resource stretch $Z$ is the same for Cost266 dataset and ta2 dataset, the actual amount of resources is different because the maximum flow rate of ta2 dataset is 140 times more than that of the cost255 dataset. However, the total flow rates of ta2 dataset are 50 times less than the total flow rates of Cost266 dataset. That explains why for a similar budget, we have a better performance for all algorithms under ta2 dataset (e.g., Fig.~\ref{fig:otherTopologies}(\subref{fig:ta2_TotalFlow_10_1869includeAll}))  compared to Cost266 dataset (e.g., Fig.~\ref{fig:otherTopologies}(\subref{fig:cost266_TotalFlow_10_1332includeAll})). 
    
    \blue{In addition, we also study the impact on the resource utilization of VNF-nodes. For each VNF-node, we compute the utilization per resource type as the amount of used resource divided by the total amount of each available resource, which is then averaged over all resource types and over all VNF-nodes. We compare our proposed algorithms with the optimal solution and present these results in Fig.~\ref{fig:utilization}. The SG-PRA and SG-NRA algorithms exhibit similar results to our algorithms as the utilization primarily depends on the resource allocation component, which is the same for both the SG-PRA/SG-NRA algorithms and our proposed algorithms. We can observe that the utilization increases as the resource stretch increases, mainly due to the increase in the total processed traffic as shown in Figs.~\ref{fig:otherTopologies}(\subref{fig:cost266_TotalFlow_10_1332includeAll}) and \ref{fig:otherTopologies}(\subref{fig:ta2_TotalFlow_10_1869includeAll}). However, the utilization naturally decreases once close to 100\% of the traffic has been processed (see Fig.~\ref{fig:utilization}(\subref{fig:utilization_ta2_10nodes})). This is because more resource becomes available but not used.}
    
    \blue{Finally, we present the average running time of all algorithms in Table~\ref{table:running_time}. We observe that all algorithms have a small running time; even for the ILP, solved using the Gurobi solver, the running time is small as well. As mentioned earlier, while the ILP solver seems to work reasonably well for the problem instances considered in our simulations, there is no guarantee that the ILP can be efficiently solved due to its NP-hardness. Therefore, we use the it as a benchmark for comparison purposes only.}
    
    \begin{figure}[t]
        \centering
        \begin{subfigure}[t]{0.48\linewidth}
            \centering
            \includegraphics[width=0.99\linewidth]{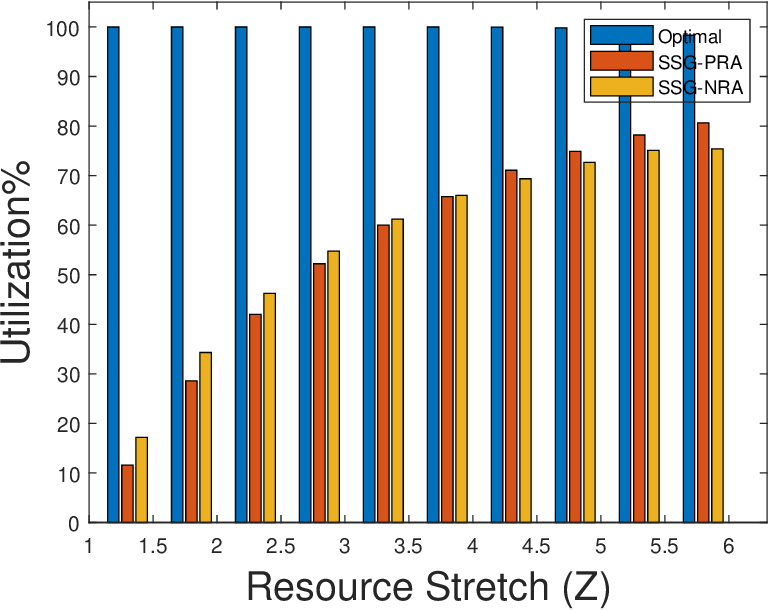}
            \caption{Cost266 \vspace{1mm}}
            \label{fig:utilization_cost266_10nodes}
        \end{subfigure}
        \begin{subfigure}[t]{0.48\linewidth}
            \centering
            \includegraphics[width=0.99\linewidth]{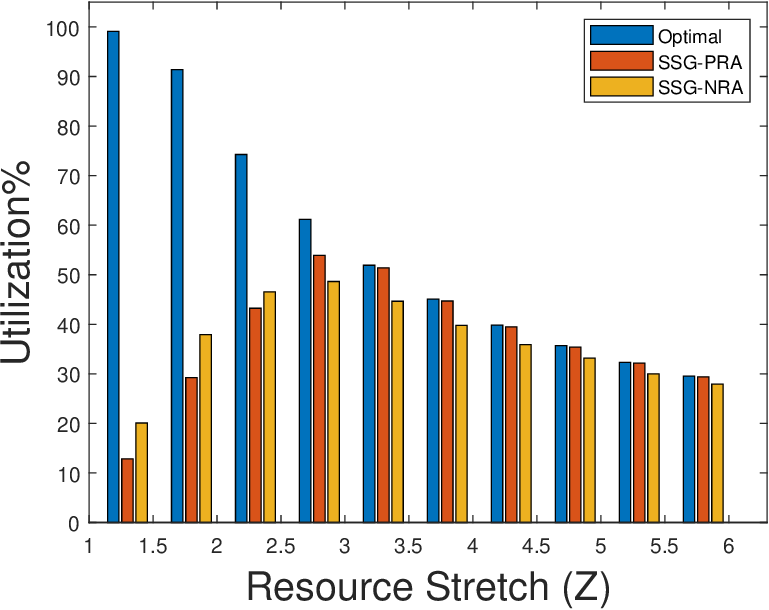}
            \caption{Ta2 \vspace{1mm}}
            \label{fig:utilization_ta2_10nodes}
        \end{subfigure}
         \caption{\blue{VNF-node utilization with budget = 10 VNF-nodes}}
        \label{fig:utilization}
    \end{figure}
    
    \begin{table*}[t]
        \centering
        \begin{tabular}{| c | c |c| c | c| c|}
            \hline
            \backslashbox{Dataset}{Algorithm} & ILP & SSG-PRA & SSG-NRA & SG-PRA & SG-NRA \\[1ex]  \hline
            cost266 & 17 & 8 & 8 & 9 & 9 \\ [1ex]  \hline
            ta2 & 10 & 12 & 11 & 12 & 12 \\ [1ex]  \hline
        \end{tabular}
     \caption{\blue{Running time (seconds)}}
        \label{table:running_time}
    \end{table*}
    
    \section{Related Work}
    The transition to NFV often happens in two phases: the planning phase and the production phase. In the planning phase, the main focus is on deciding where to introduce the NFV capabilities to maximize the benefit (i.e., the placement problem). Since this phase happens before the production phase, we can assume certain behaviors about the production phase and use some historical traces to project flow demands and routing. In the production phase, finer-grained optimization will be considered, such as flow admission and routing. In this work, we focus on the planning phase and assume a fixed routing. In the following, we will discuss related work on several relevant topics.
    
    \textbf{Placement:} The placement problem has been considered in different domains such as NFV (e.g., \cite{sallam2021joint}), software-defined networks (SDN) (e.g., \cite{poularakis2017one}), and edge cloud computing (e.g., \cite{He2018It}). In NFV, several studies (e.g., \cite{sang2017provably,chen2018virtual, tomassilli2018provably}) consider the placement of a minimum number of VNF instances to cover all flows. A single type of network functions is considered in  \cite{sang2017provably,chen2018virtual, shi2018competitive, lukovszki2015online}, and the case of multiple network functions is considered in \cite{tomassilli2018provably, ren2019embedding, sallam2018shortest, Feng2017, woldeyohannes2018cluspr, poularakis2020approximation, gu2020service}. However, these work neglects either the budget constraint or the multi-dimensional resource allocation. The work in \cite{lukovszki2018approximate} considers the placement of middleboxes to make the shortest path between communicating pairs under a threshold. A generalized version of \cite{lukovszki2018approximate} is presented in \cite{poularakis2020service}, which considers the problem of joint service placement and request routing in mobile edge computing networks. Again, this work does not consider multiple network functions or budget constraint. Other work considers different objectives, such as delay minimization (e.g. \cite{sun2020energy, ren2020efficient}), \blue{energy efficiency (e.g., \cite{zhang2020energy, soualah2017energy})}, fault tolerance (e.g., \cite{yuan2020fault}), \blue{and revenue maximization (e.g., \cite{nejad2018vspace, golkarifard2021dynamic})}. 
    
    \blue{\textbf{Budget-constrained Resource Allocation:} Budget-constrained resource allocation has been considered in prior studies (e.g., \cite{xu2020qos, poularakis2017one, He2018It, farhadi2021service, sallam2021joint}). In \cite{xu2020qos}, the authors consider the problem of request routing and capacity allocation with the objective of maximizing the total traffic of admitted requests under a limited budget. 
    They assume that all nodes in the network already have NFV capabilities and focus on maximizing the utilization of the budget in the case that admitting a request has a certain cost. In the SDN domain, the work in \cite{poularakis2017one} considers the placement of SDN-enabled routers to maximize the total processed traffic. They consider a budget constraint but neglect the resource constraint. Similarly, in the work on edge cloud computing \cite{He2018It, farhadi2021service}, although the budget and resource constraints are considered, their proposed solution is only for a special case, and the overall problem does not consider the multi-dimensional setting. In our previous work \cite{sallam2021joint}, we study the NFV placement problem where budget, resource, and fully processed flow constraints are considered, but we only consider one type of network function and one type of resource.}

    \textbf{Multi-dimensional Resource Allocation:} To the best of our knowledge, the multi-dimensional setting has rarely been considered except in a limited number of studies. In \cite{yu2018ensc}, the authors consider multi-resource VNFs with a focus on the analysis of the vertical scaling (scaling up/down of some resources) and horizontal scaling (the number of VNFs instances). The work of \cite{even2013competitive} focuses only on request admission and routing. The work of \cite{wang2017multi} also considers the multi-resource setting, but the focus is on how to balance the load among the servers, taking into consideration the different demand of network functions for each resource. In \cite{luo2020online}, the authors consider servers with two types of resources, and their objective is to serve all network flows by placing network functions on these servers with minimum cost. Our work considers all three constraints of budget, resource, and fully processed flows, as well as the multi-dimensional setting. 
    
    \textbf{Sequence Submodularity:} The concept of sequence (or string) submodularity is a generalization of submodularity, which has been recently introduced in several studies (e.g., \cite{streeter2009online, zhang2016string, alaei2010maximizing, tschiatschek2017selecting}). It models objective functions that depend on the sequence of actions. It has been shown in \cite{streeter2009online} that a simple greedy algorithm can achieve an approximation ratio of $(1 - 1/e)$ for maximizing forward-monotone, backward-monotone, and sequence-submodular functions. To the best of our knowledge, we are the first to utilize the concept of sequence submodularity for the placement problem in NFV. Although in \cite{alaei2010maximizing}, the backward-monotone property is not assumed, per our investigation, there is a crucial step in the proof that requires this property. However, the objective function of our 
    placement subproblem (i.e., Problem \eqref{eq:sequencePlacement} defined in Section \ref{subsec:second_relaxation}) does not satisfy this backward-monotone property, rendering the results of \cite{alaei2010maximizing} inapplicable to our problem. To address this new challenge, we introduce the concept of approximate backward-monotone and extend the previous result to this generalization.
    
    \blue{\textbf{Preliminary Version:} A preliminary version of this work was presented at IEEE ICNP 2019 \cite{sallam2019placement}. This extended journal version includes a correction of Lemma 4 in \cite{sallam2019placement} (Theorem \ref{theorem:multiVPRAPlacement} in this version). In \cite{sallam2019placement}, we applied the results in \cite[Theorem 3]{alaei2010maximizing}, which states that for an objective function that is forward-monotone and sequence-submodular, the greedy algorithm (Algorithm~\ref{alg:SSG}) achieves an approximation ratio of $(1 - 1/e)$. However, it turns out that we also need the objective function to satisfy a backward-monotone property (defined in Section \ref{subsec:sequence_submodular}) in order to have the $(1 - 1/e)$ approximation. This is consistent with the results in other highly relevant work (e.g., \cite{streeter2009online}). We show that our objective function does not satisfy the backward-monotone property but satisfies an approximate version of the property. Based on that, we derive the approximation ratio of the greedy algorithm for forward-monotone, approximate backward-monotone, and sequence-submodular functions and apply the result to our problem.}

    \section{\blue{Conclusion and Future Work}}
    In this paper, we considered the problem of placement and resource allocation of VNF-nodes. We showed that considering the multi-dimensional setting along with the budget, limited resources, and fully flow processing constraints introduces several new challenges. However, through a two-level relaxation, we were able to develop an efficient placement algorithm. In addition, we utilized the primal-dual technique to design efficient resource allocation algorithms that properly handle the multi-dimensional setting. Although the second-level relaxation results in a smaller approximation ratio (a factor of 1/2), we showed through simulation that its impact of the empirical performance is negligible. Besides, the simulation results agree with the derived approximation ratio of both resource allocation algorithms. Specifically, the simulation showed that for a smaller resource stretch $Z$ and larger number of nodes, the NRA algorithm works better; when $Z$ becomes large enough, the PRA algorithm is better than the NRA algorithm and becomes close to the optimal solution earlier. 
    
    \blue{In our future work, we would like to consider the following interesting directions. First, we will consider service function chaining, where the network functions required for each flow must be in a specific order. Second, we will also consider heterogeneous node cost. It is worth noting that in our previous work \cite{sallam2021joint}, we have considered heterogeneous node cost where a single type of resource and a single type of network function are assumed, and we were able to develop efficient algorithms with approximation ratio guarantees. However, it is unclear whether one can extend such algorithms to the multi-dimensional setting and still achieve certain performance guarantees. Finally, it would be interesting to explore the possibility of developing polynomial-time approximation schemes (PTAS) for the resource allocation subproblem and the overall problem we consider.} 
    
    \section{Proofs}
\subsection{Proof of Lemma \ref{lemma:prefixequal}}
\label{proof:prefixequal}
 \begin{proof}
       Based on Algorithm \ref{alg:fractional}, nodes in $S_3$ will be considered first when solving for sequence $S_1$ and $S_2$, and in the same order. Therefore, the amount of assigned traffic to nodes in $S_3$ will be the same.
  \end{proof}
\subsection{Proof of Lemma \ref{lemma:v_from_lambda_S}}
\label{proof:v_from_lambda_S}

  \begin{proof}
  First, we can express the right-hand side of Eq. \eqref{eq:v_from_lambda_S} as follows:
  \begin{equation*}
      \begin{aligned}
        & R_4(S_1 \oplus u| \boldsymbol{\lambda}) - R_4(S_1 | \boldsymbol{\lambda}) \\
        = & \sum_{i = 1}^{|S_1|} \hat{y}_i(S_1 \oplus u) +  \hat{y}_{|S_1 \oplus u|}(S_1 \oplus u) - \sum_{i = 1}^{|S_1|} \hat{y}_i(S_1) \\
        = &  \hat{y}_{|S_1 \oplus u|}(S_1 \oplus u),
      \end{aligned}
  \end{equation*}
  where the last equality holds because $\hat{y}_i(S_1 \oplus u) = \hat{y}_i(S_1)$ for $i \leq |S_1|$ from Lemma \ref{lemma:prefixequal}. Recall from Algorithm \ref{alg:fractional} that $\hat{y}_{|S_1 \oplus u|}(S_1 \oplus u)$ is the value of the optimal solution of Problem \eqref{eq:nodeCapacityAllocation} for node $u$ given $[\hat{y}_1(S_1 \oplus u), \dots, \hat{y}_{|S_1|}(S_1 \oplus u), 0]$, which we include in the following to easily navigate the proof:
    \begin{subequations}
    \label{eq:nodeCapacityAllocation_2}
    \begin{align}
    & \underset{\boldsymbol{X}(S_1 \oplus u)}{\text{maximize}} \quad  \sum_{f\in \F}  x_f^{u} \nonumber\\
    & \text{subject to} \nonumber \\
    & \sum_{v \in \V_f \cap \V(S_1 \oplus u)} x_f^v \leq \lambda_f, && \forall f \in \F \text{ and } v \in \V_f, \label{const:a}\\
    & x_f^v=0, && \forall f \in \F \text{ and } v \notin \V_f, \label{const:b}\\
    & \sum_{\phi \in \Phi} \beta_\phi^r \sum_{f \in \F(\phi)} x_f^{v_i} \leq c_{v_i}^r, && \forall  r \in \R,~  i=1, \dots, |S_1 \oplus u|, \label{const:c}\\
    & \sum_{f \in \F} x_f^{v_i} =  \hat{y}_i(S_1 \oplus u), && i = 1, \dots, |S_1|. \label{const:d}
    \end{align}
\end{subequations}

    On the other hand, $R_4(u | (\boldsymbol{\lambda} - \hat{\boldsymbol{x}}(S_2)))$ corresponds to solving the following problem:
     \begin{equation}
    \label{eq:node_u_capacity_allocation}
    \begin{aligned}
    & \underset{\boldsymbol{X}(u)}{\text{maximize}} \quad  \sum_{f\in \F}  x_f^{u}\\
    & \text{subject to} \\
    & x_f^u \leq \lambda_f - \hat{x}_f(S_2), && \forall f \in \F \text{ and } u \in \V_f \\
    & x_f^u=0, && \forall f \in \F \text{ and } u \notin \V_f, \\
    & \sum_{\phi \in \Phi} \beta_\phi^r \sum_{f \in \F(\phi)} x_f^{u} \leq c_{u}^r, && \forall  r \in \R.
    \end{aligned}
    \end{equation}

    Recall that $\hat{\boldsymbol{X}}(S_2)$ denotes the flow assignment matrix of sequence $S_2$ at the end of Algorithm \ref{alg:fractional}; we also use $\hat{\boldsymbol{X}}(u)$ to denote the flow assignment matrix of sequence $(u)$ after solving Problem \eqref{eq:node_u_capacity_allocation}. We construct a flow assignment matrix $\bar{\boldsymbol{X}}(S_1 \oplus u)$ by concatenating the flow assignment of the first $|S_1|$ nodes in $\hat{\boldsymbol{X}}(S_2)$ with the flow assignment of node $u$ in $\hat{\boldsymbol{X}}(u)$, i.e., we let $\bar{\boldsymbol{X}}(S_1 \oplus u) = [ \hat{\boldsymbol{X}}(S_2)_{(:, 1:|S_1|)},~ \hat{\boldsymbol{X}}(u)]$.  It is easy to see that the value of the constructed flow assignment matrix $\bar{\boldsymbol{X}}(S_1 \oplus u)$ is equal to $R_4(u | (\boldsymbol{\lambda} - \hat{\boldsymbol{x}}(S_2)))$. To conclude the lemma, we need to show that $\bar{\boldsymbol{X}}(S_1 \oplus u)$ is a feasible solution to Problem \eqref{eq:nodeCapacityAllocation_2}. We do so by showing that each constraint of Problem \eqref{eq:nodeCapacityAllocation_2} is satisfied as follows.
    
    1) Constraint \eqref{const:a}. For each flow $f$, we have
    \begin{equation*}
        \begin{aligned}
          & \sum_{v \in \V_f \cap \V(S_1 \oplus u)}  \bar{x}_f^v(S_1 \oplus u) \\
        \stackrel{\text{(a)}} = & \sum_{i =1}^{|S_1|} \bar{x}_f^{v_i}(S_1 \oplus u) +  \bar{x}_f^u(S_1 \oplus u) \\
       \stackrel{\text{(b)}} = &  \sum_{i =1}^{|S_1|} \hat{x}_f^{v_i}(S_2) + \hat{x}_f^u((u)) \\
         \stackrel{\text{(c)}} \leq & \sum_{i =1}^{|S_1|} \hat{x}_f^{v_i}(S_2) +  \lambda_f - \hat{x}_f(S_2) \\
         \stackrel{\text{(d)}} = & \sum_{i =1}^{|S_1|} \hat{x}_f^{v_i}(S_2) +  \lambda_f - \sum_{i =1}^{|S_2|} \hat{x}_f^{v_i}(S_2) \\
        \stackrel{\text{(e)}} \leq & \lambda_f,
        \end{aligned}
    \end{equation*}
   where (a) holds because $\bar{x}_f^v(S_1 \oplus u) = 0$ for $v \notin \V_f$; (b) follows from the way we constructed $\bar{\boldsymbol{X}}(S_1 \oplus u)$; (c) holds because $\hat{x}_f^u(u) \leq \lambda_f - \hat{x}_f(S_2)$ from Problem \eqref{eq:node_u_capacity_allocation}; (d) holds from the definition of $\hat{x}_f(S_2)$ in Eq. \eqref{eq:gamma_j_S}; (e) holds because $|S_1| \leq |S_2|$.
   
   2) Constraint \eqref{const:b}. This constraint is satisfied by the feasibility of $\hat{\boldsymbol{X}}(S_2)$ and $\hat{\boldsymbol{X}}(u)$.
   
   3) Constraint \eqref{const:c}. The flow assignment matrix $\hat{\boldsymbol{X}}(S_2)_{(:, 1:|S_1|)}$ satisfies Constraint \eqref{const:c} for sequence $S_1$ and does not assign any traffic to node $u$. Similarly, the flow assignment matrix $\hat{\boldsymbol{X}}(u)$ satisfies Constraint \eqref{const:c} for node $u$ without affecting the assigned traffic to sequence $S_1$. Therefore, the constructed flow assignment matrix $\bar{\boldsymbol{X}}(S_1 \oplus u)$ satisfies Constraint \eqref{const:c} for sequence $S_1 \oplus u$.
   
   4) Constraint \eqref{const:d}. For $i =1, \dots, |S_1|$, we have
   \begin{equation*}
        \begin{aligned}
       \sum_{f \in \F} \bar{x}_f^{v_i}(S_1 \oplus u) &  \stackrel{\text{(a)}} = \sum_{f \in \F} \hat{x}_f^{v_i}(S_2) \\
       & \stackrel{\text{(b)}} = \hat{y}_{i}(S_2) \\
       & \stackrel{\text{(c)}} = \hat{y}_i(S_1 \oplus u),
        \end{aligned}
    \end{equation*}
   where (a) follows from the way we constructed $\bar{\boldsymbol{X}}(S_1 \oplus u)$; (b) holds from the feasibility of $\hat{\boldsymbol{X}}(S_2)_{(:, 1:|S_1|)}$; (c) holds because  $\hat{y}_i(S_2) = \hat{y}_i(S_1 \oplus u)$ for $i \leq |S_1|$ from Lemma \ref{lemma:prefixequal}. 
   
  \end{proof}
  
\subsection{Proof of Lemma \ref{lemma:setToSequence}}
\label{proof:setToSequence}
\begin{proof}
Define $\hat{\boldsymbol{x}}^{\text{min}}(\U) \triangleq \min(\hat{\boldsymbol{x}}(\U), \hat{\boldsymbol{x}}(S))$ and $\hat{\boldsymbol{x}}^{\text{max}}(\U) \triangleq \max(\hat{\boldsymbol{x}}(\U) - \hat{\boldsymbol{x}}(S), 0)$. Note that $\hat{\boldsymbol{x}}^{\text{min}}(\U) + \hat{\boldsymbol{x}}^{\text{max}}(\U) = \hat{\boldsymbol{x}}(\U)$. To show that $\frac{1}{2}R_3(\U) \leq R_4(S)$, we prove the following: 
\begin{enumerate}[label=(\Alph*)]
    \item $\sum_{j = 1}^{F} \hat{x}_{j}^{\text{min}}(\U) \leq R_4(S)$;
    \item $\sum_{j = 1}^{F} \hat{x}_{j}^{\text{max}}(\U) \leq R_4(S)$.
\end{enumerate}
By combining (A) and (B), we get that
\begin{equation*}
    \begin{aligned}
    R_3(\U) & = \sum_{j = 1}^{F} \hat{x}_{j}(\U) \\
    & = \sum_{j = 1}^{F} \hat{x}_{j}^{\text{min}}(\U) + \sum_{j = 1}^{F} \hat{x}_{j}^{\text{max}}(\U) \\
    & \leq 2 R_4(S).
    \end{aligned}
\end{equation*}

To show that (A) holds, we have
\begin{equation}
    \begin{aligned}
    \sum_{j = 1}^{F} \hat{x}_{j}^{\text{min}}(\U) & \stackrel{\text{(a)}}{\leq} \sum_{j = 1}^{F} \hat{x}_{j}(S)  \stackrel{\text{(b)}} = R_4(S),
    \end{aligned}
\end{equation} 
where (a) holds since $\hat{\boldsymbol{x}}^{\text{min}}(\U) \leq \hat{\boldsymbol{x}}(S)$; (b) holds from the definition of function $R_4$ in Eq. \eqref{eq:R_4_S_flows_based}.

Next, we show that (B) holds. Recall that by definition, the flow rate vector $\hat{\boldsymbol{x}}(\U)$ can be fully assigned to nodes $\U$. Therefore, any vector $\boldsymbol{\bar{x}}(\U) \leq \hat{\boldsymbol{x}}(\U)$ can be fully assigned to nodes $\U$ as well, and we can establish the following:
\begin{equation}
    \label{eq:U_from_smaller_flow_vector}
    R_3(\U | \boldsymbol{\bar{x}}(\U)) = \sum_{j = 1}^{F} \bar{x}_{j}(\U).
\end{equation}
 
Moreover, we have
   \begin{equation}
        \begin{aligned}
   \sum_{j = 1}^{F} \hat{x}_{j}^{\text{max}}(\U) & \stackrel{\text{(a)}} = R_3(\U | \hat{\boldsymbol{x}}^{\text{max}}(\U))\\
   &  \stackrel{\text{(b)}} \leq  R_3(\U | (\boldsymbol{\lambda} - \hat{\boldsymbol{x}}(S)))\\
        & \stackrel{\text{(c)}} \leq  \sum_{v \in \U} R_3(\{v\} | (\boldsymbol{\lambda} - \hat{\boldsymbol{x}}(S)))\\
        & \stackrel{\text{(d)}} =  \sum_{i = 1}^{|S|} R_4(v_i | (\boldsymbol{\lambda} - \hat{\boldsymbol{x}}(S)))\\
        & \stackrel{\text{(e)}} \leq  \sum_{i = 1}^{|S|} (R_4((v_1, \dots, v_{i-1}, v_i) |\boldsymbol{\lambda}) \\
        & - R_4((v_1, \dots, v_{i-1}) |\boldsymbol{\lambda}))\\
        & = R_4(S),
       \end{aligned}
    \end{equation} 
    where (a)-(e) hold for the following. (a) follows from Eq. \eqref{eq:U_from_smaller_flow_vector}. For (b), recall that $\hat{\boldsymbol{x}}^{\text{max}}(\U) = \max(\hat{\boldsymbol{x}}(\U) - \hat{\boldsymbol{x}}(S), 0)$. Since $\hat{\boldsymbol{x}}(\U) \leq \boldsymbol{\lambda}$, we get that $\hat{\boldsymbol{x}}^{\text{max}}(\U) \leq \boldsymbol{\lambda} - \hat{\boldsymbol{x}}(S)$. By applying Eq. \eqref{eq:compare_two_vectors_R_3}, we get that (b) holds. (c) holds because we consider each node in $\U$ individually with the same traffic rate vector $(\boldsymbol{\lambda} - \hat{\boldsymbol{x}}(S))$, so the resource allocation in the solution to nodes $\U$ is a feasible solution to each individual node. For (d), note that when applying Algorithm \ref{alg:fractional} for a singleton sequence, the equality constraints of Problem \eqref{eq:nodeCapacityAllocation} are irrelevant and function $R_4((v_i))$ and $R_3(\{v_i\})$ will get the same result, so (d) holds. (e) holds from Lemma \ref{lemma:v_from_lambda_S}, where $(v_1, \dots, v_{i-1}) \preceq S$ and node $v_i \notin \V((v_1, \dots, v_{i-1}))$.

Finally, we show that $R_4(S) \leq R_3(\U)$. We can see that the flow assignment matrix $\hat{\boldsymbol{X}}(S)$ satisfies the flow rate constraint of all flows and the resources constraint of nodes in $S$. Since $S \in \mathcal{P}(\U)$, then $\hat{\boldsymbol{X}}(S) \in \mathcal{X}(\U)$, and we get that $\hat{\boldsymbol{X}}(S)$ is a feasible solution to Problem \eqref{eq:relaxedAllocation}. Therefore, the inequality holds. 
   \end{proof}

\subsection{Proof of Theorem \ref{theorem:SSG}}
\label{proof:SSG}
    \begin{proof}
         First, we present Lemma \ref{lemma:node_marginal_value}, which shows that adding a node according to Line 2 of  Algorithm \ref{alg:SSG} yields a marginal value that is greater than or equal to the average of adding any other sequence of nodes. Let $S^{i} = (v_1, \dots, v_i)$ denote the sequence of nodes selected by Algorithm \ref{alg:SSG} in the first $i$ iterations, with $S^0$ to denote an empty sequence. Also, for any two sequences $S_1$ and $S_2$, we define $h(S_2|S_1) \triangleq h(S_1 \oplus S_2) - h(S_1)$.
         
     \begin{lemma}[\cite{alaei2010maximizing}]
        \label{lemma:node_marginal_value}
       Let $v_i$ denote the node selected in iteration $i$ by Algorithm \ref{alg:SSG}. For any sequence $S^\prime \in \HH$, we have $h(v_i|S^{i-1}) \geq \frac{1}{|S^\prime|} h(S^\prime|S^{i-1})$.
    \end{lemma}
 
    Next, let $S^*$ denote the optimal sequence. Since function $h$ is forward-monotone, we can assume that the length of sequence $S^*$ is exactly $k$. Using Lemma \ref{lemma:node_marginal_value}, we have
    \begin{equation}
    \begin{aligned}
    \label{eq:node_marginal_value}
        h(v_i|S^{i-1}) & \geq \frac{1}{k} h(S^*|S^{i-1}) \\
        & = \frac{1}{k} (h(S^{i-1} \oplus S^*) - h(S^{i-1})) \\
        & \geq \frac{1}{k} (\alpha h(S^*) - h(S^{i-1})),
    \end{aligned}
    \end{equation}
    where  the last inequality holds because function $h$ is $\alpha$-backward-monotone. By rewriting Eq. \eqref{eq:node_marginal_value}, we have
    \begin{equation}
        h(S^{i}) - h(S^{i-1})  \geq \frac{1}{k} (\alpha h(S^*) - h(S^{i-1})),
    \end{equation}
    which is equivalent to 
     \begin{equation}
     \label{eq:recursive_i}
        h(S^{i})  \geq \frac{\alpha}{k} h(S^*) + (1-\frac{1}{k})h(S^{i-1})),
    \end{equation}
    Writing Eq. \eqref{eq:recursive_i} for $i = k$ and expanding it yields
    \begin{equation}
     \begin{aligned}
     \label{eq:recursive_k}
        h(S^{k})  \geq & \frac{\alpha}{k} h(S^*)  \\ 
        & + \frac{\alpha}{k}(1-\frac{1}{k})h(S^*) \\ 
        & + \frac{\alpha}{k}(1-\frac{1}{k})^2 h(S^*) \\
        & + \dots \\
        & + \frac{\alpha}{k}(1-\frac{1}{k})^{k-1} h(S^*)+ (1-\frac{1}{k})^k h(S^{0})) \\
        & = \alpha(1 - (1 - \frac{1}{k})^k) h(S^*) \\
        & \geq \alpha(1 - \frac{1}{e}) h(S^*),
     \end{aligned}
    \end{equation}
    where the last inequality holds because $(1 - \frac{1}{k})^k \leq \frac{1}{e}$.
    \end{proof}

\subsection{Proof of Lemma \ref{lemma:submodularity}}
\label{proof:submodularity}
    \begin{proof}
    First, for any two sequences $S_1, S_2,$ we assume that $S_1 \oplus S_2$ has no repeated nodes, because if otherwise, we can remove the later appearance of the same node without affecting the value of $R_4(S_1 \oplus S_2)$.
    
    Now, we proceed with the proof. First, we show that function $R_4(S)$ is forward-monotone (i.e., satisfies Eq. \eqref{eq:forward_monotone}). Since $S_1 \preceq S_1 \oplus S_2$, then according to Lemma \ref{lemma:prefixequal}, the assigned traffic to nodes in $S_1$ will be the same for the two sequences $S_1$ and $S_1 \oplus S_2$. Adding additional nodes to sequence $S_1$ will not affect the amount of traffic already assigned to nodes in $S_1$, and the minimum that can be assigned to any node is zero. So, Eq. \eqref{eq:forward_monotone} is satisfied.
    
    Next, we show that function $R_4(S)$ is $\frac{1}{2}$-backward-monotone (i.e., satisfies Eq. \eqref{eq:alpha_backward_monotone} with $\alpha = \frac{1}{2}$).  We have the following:
         \begin{equation}
        \begin{aligned}
        \label{eq:backward_bound}
            R_4(S_1 \oplus S_2) & \stackrel{\text{(a)}} \geq \frac{1}{2} R_3(\V(S_1 \oplus S_2)) \\
             & \stackrel{\text{(b)}} \geq \frac{1}{2} R_3(\V(S_2)) \\
             & \stackrel{\text{(c)}} \geq \frac{1}{2} R_4(S_2)
        \end{aligned}
    \end{equation} 
    where (a) follows from Lemma \ref{lemma:setToSequence}. For (b), note that function $R_3$ is monotonically nondecreasing because adding an additional VNF-node does not reduce the amount of flows that can be processed. Since $\V(S_2)$ is a subset of $\V(S_1 \oplus S_2)$ and function $R_3$ is nondecreasing, then (b) holds. (c) holds from Lemma \ref{lemma:setToSequence}.

     Finally, we show that  function $R_4(S)$ is sequence-submodular (i.e., satisfies Eq. \eqref{eq:sequence_submodular}). For  Eq. \eqref{eq:sequence_submodular} to be satisfied, we need to show that $R_4(S_1 \oplus u) -  R_4(S_1) \geq R_4(S_2 \oplus u) -  R_4(S_2)$ for any $S_1 \preceq S_2$ and $u \in \V$. We distinguish between two cases: 
     
     Case I: node $u \in \V(S_2)$. In this case, according to the definition of function $R_4$, a repeated node has zero marginal gain (i.e., $R_4(S_2 \oplus u) -  R_4(S_2) = 0$). Since function $R_4$ is forward-monotone, then we also get that $R_4(S_1 \oplus u) -  R_4(S_1) \geq 0$. Therefore, function $R_4$ is sequence-submodular.
     
     Case II: node $u \notin \V(S_2)$, which also implies that $u \notin \V(S_1)$. Let $\hat{\boldsymbol{X}}(S_2 \oplus u)$ denote the flow assignment matrix of sequence $S_2 \oplus u$ at the end of Algorithm \ref{alg:fractional}. Let $\hat{\boldsymbol{x}}(S_2 \oplus u)$ be the flow rate vector extracted from $\hat{\boldsymbol{X}}(S_2 \oplus u)$. We define $\bar{\boldsymbol{X}}(S_2) \triangleq [ \hat{\boldsymbol{X}}(S_2 \oplus u)_{(:, 1:|S_2|)}]$. Note that $\bar{\boldsymbol{X}}(S_2)$ is a possible realization of the flow assignment matrix of sequence $S_2$ at the end of Algorithm \ref{alg:fractional}. Let $\bar{\boldsymbol{x}}(S_2)$ be the flow rate vector extracted from $\bar{\boldsymbol{X}}(S_2)$. The proof proceeds as follows:
     \begin{equation}
        \begin{aligned}
           & R_4(S_2 \oplus u) - R_4(S_2)  \\
           = & \sum_{i = 1}^{|S_2|} \hat{y}_i(S_2 \oplus u) + \hat{y}_{|S_2 \oplus u|}(S_2 \oplus u) - \sum_{i = 1}^{|S_2|} \hat{y}_i(S_2)\\
             \stackrel{\text{(a)}} = & \hat{y}_{|S_2 \oplus u|}(S_2 \oplus u)\\
             = & \sum_{j = 1}^{F} \hat{x}_j(S_2 \oplus u) - \bar{x}_j(S_2) \\
             \stackrel{\text{(b)}} = & R_4(u | (\hat{\boldsymbol{x}}(S_2 \oplus u) - \bar{\boldsymbol{x}}(S_2))) \\
            \stackrel{\text{(c)}} \leq & R_4(u | (\boldsymbol{\lambda} - \bar{\boldsymbol{x}}(S_2))) \\
            \stackrel{\text{(d)}} \leq & R_4(S_1 \oplus u) - R_4(S_1),
        \end{aligned}
    \end{equation} 
     where (a) holds because $\hat{y}_i(S_2 \oplus u) = \hat{y}_i(S_2)$ for $i \leq |S_2|$ from Lemma \ref{lemma:prefixequal}; for (b), the flow rate vector $(\hat{\boldsymbol{x}}(S_2 \oplus u) - \bar{\boldsymbol{x}}(S_2))$ is what has been assigned to node $u$ while considering sequence $S_2 \oplus u$, so it can also be assigned to node $u$ when considering the singleton sequence consisting of only node $u$; (c) holds from Eq. \eqref{eq:compare_two_vectors_R_4} since $\hat{\boldsymbol{x}}(S_2 \oplus u) \leq \boldsymbol{\lambda}$; (d) holds from Lemma \ref{lemma:v_from_lambda_S} since $S_1 \preceq S_2$ and node $u \notin \V(S_1)$.
    \end{proof}
    
\subsection{Function $R_4$ $\frac{1}{2}$-backward-monotone: Tight Bound Proof}
\label{example:lower_bound}
    We show that the lower bound of $\frac{1}{2}$ in Eq. \eqref{eq:backward_bound} is tight through the following problem instance. Consider three nodes $v_1, v_2, v_3$, three flows $f_1, f_2, f_3$, and two types of resources $r_1, r_2$. Assume the following: the traffic rate of each flow is $\epsilon_1, ~ 1,$ and $1+\epsilon_2$, respectively, with $\epsilon_1 > \epsilon_2$ for arbitrary small $\epsilon_1$ and $\epsilon_2$; the amount of each resource at each node is $c$; the path of each flow is $(v_1, v_2),~ (v_2),$ and $(v_2, v_3)$, respectively. Let $\delta_f^r \triangleq \sum_{\phi \in  \Phi_f} \beta_\phi^r$ be the total amount of resource $r$ needed to process a unit of flow $f$ by the set of network functions $\Phi_f$. Assume that the following holds: 
    \begin{enumerate}[label=\roman*)]
        \item  $\delta_{f_1}^{r_1} \times \lambda_{f_1} ~ + ~ \delta_{f_2}^{r_1} \times \lambda_{f_2} = c$, \vspace{3pt}
        \item  $\delta_{f_1}^{r_2} \times \lambda_{f_1} ~ + ~ \delta_{f_2}^{r_2} \times \lambda_{f_2} = c$,\vspace{3pt}
        \item $\delta_{f_3}^{r_1} \times \lambda_{f_3} = \delta_{f_3}^{r_2} \times \lambda_{f_3} = c$,  \vspace{3pt}
        \item $\delta_{f_2}^{r_1} > \delta_{f_3}^{r_1}$.
    \end{enumerate} 
    It can be verified that if the above assumptions hold, then $\boldsymbol{y}((v_1, v_2, v_3)) = [\epsilon_1, 1+\epsilon_2, 0]$, while $\boldsymbol{y}((v_2, v_3)) = [1 + \epsilon_1, 1+ \epsilon_2]$. As a result, we get that $R_4(v_1, v_2, v_3) ~=~ 1+\epsilon_1 + \epsilon_2$ and $R_4(v_2, v_3) ~=~ 2+ \epsilon_1 + \epsilon_2$, for arbitrary small $\epsilon_1$ and $\epsilon_2$. 
    
     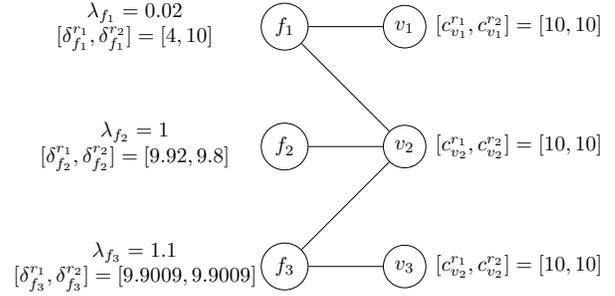
\begin{figure}
        \centering
        \resizebox{0.9\linewidth}{!}{
            \begin{tikzpicture}[transform shape]
            
            \node[vertex](f_1) at (0, 2) {$ f_1 $};
            \node[vertex](f_2) at (0, 0){$ f_2 $};
            \node[vertex ](f_3) at (0, -2){$ f_3 $};
            
            \node[vertex](v_1) at (2, 2) {$ v_1 $};
            \node[vertex](v_2) at (2, 0){$ v_2 $};
            \node[vertex ](v_3) at (2, -2){$ v_3 $};

            \begin{scope}[every path/.style={-}, every node/.style={inner sep=1pt}]

            \path (f_1) edge [ anchor= east] node {} (v_1);
            \path (f_1) edge [ anchor= south] node {} (v_2);
            \path (f_2) edge [ anchor= south] node {} (v_2); 
            \path (f_3) edge [ anchor= south] node {} (v_2);
            \path (f_3) edge [ anchor= south] node {} (v_3); 
            
            \node[align=center] at (-2.5,2) {$\lambda_{f_1} = 0.02$ \\ $[\delta_{f_1}^{r_1}, \delta_{f_1}^{r_2}] = [4, 10]$};
            
            \node[align=center] at (-2.5,0) {$\lambda_{f_2} = 1$ \\ $[\delta_{f_2}^{r_1}, \delta_{f_2}^{r_2}] = [9.92, 9.8]$};
            
            \node[align=center] at (-2.5,-2) {$\lambda_{f_3} = 1.1$ \\ $[\delta_{f_3}^{r_1}, \delta_{f_3}^{r_2}] = [9.9009, 9.9009]$};
            
            \node[align=center] at (3.9,2) {$[c_{v_1}^{r_1}, c_{v_1}^{r_2}] =  [10, 10]$};
            \node[align=center] at (3.9,0) {$[c_{v_2}^{r_1}, c_{v_2}^{r_2}] =  [10, 10]$};
             \node[align=center] at (3.9, -2) {$[c_{v_2}^{r_1}, c_{v_2}^{r_2}] =  [10, 10]$};
            \end{scope}
        \end{tikzpicture}}
        \caption{An example to show that the bound of $\frac{1}{2}$ in Eq. \eqref{eq:backward_bound} is tight. Computing $R_4( v_2, v_3)$ results in the following assignment vector $[1.02, 1.01]$, while computing $R_4(v_1, v_2, v_3)$ results in the following assignment vector $[0.02, 1.01, 0]$.}
        \label{fig:backward_bound}
    \end{figure}

    The following is an example (also presented in Fig. \ref{fig:backward_bound}) that satisfies the aforementioned assumptions. Consider the traffic rate of each flow to be $0.02, 1,$ and $1.01$, respectively, and the amount of each resource at each node to be $10$. Also, assume that the demand of each resource by each flow (i.e., $[\delta_{f}^{r_1}, \delta_{f}^{r_2}]$) to be $[40, 100], ~ [9.2, 8],$ and $[9.9009, 9.9009]$, respectively. First, we evaluate $R_4(v_2, v_3)$. We start with an initial resource allocation $\boldsymbol{y}((v_2, v_3))=[0, 0]$. We evaluate $R_4(v_2, v_3)$ using Algorithm \ref{alg:fractional}. We start with node $v_2$ and solve Problem \eqref{eq:nodeCapacityAllocation} for node $v_2$ given $\boldsymbol{y}((v_2, v_3))$. The algorithm will assign flow $f_1$ and $f_2$ to node $v_2$, and the total traffic assigned to node $v_2$ is equal to $1.02$. If we try to assign any combinations of flows $f_1$ and $f_3$ or of flows $f_2$ and $f_3$, the most we can assign to node $v_2$ is $1.01$. Therefore, we update $\boldsymbol{y}((v_2, v_3)) = [1.02, 0]$. Next,  we solve Problem \eqref{eq:nodeCapacityAllocation} for node $v_3$ given $\boldsymbol{y}((v_2, v_3))$. The flows assigned to node $v_2$ will remain the same and flow $f_3$ will be assigned to node $v_3$. The result is $\boldsymbol{y}((v_2, v_3))=[1.02, 1.01]$. Therefore, the value of $R_4(v_2, v_3)$ is 2.03. 
    
    Next, we evaluate $R_4(v_1, v_2, v_3)$. The initial resource allocation is $\boldsymbol{y}((v_1, v_2, v_3))=[0, 0, 0]$. We solve Problem \eqref{eq:nodeCapacityAllocation} for node $v_1$ given $\boldsymbol{y}((v_1, v_2, v_3))$. The result is that flow $f_1$ will be assigned to node $v_1$ and $\boldsymbol{y}((v_1, v_2, v_3))$ becomes $[0.02, 0, 0]$. We repeat the same steps for node $v_2$. The maximum traffic that can be assigned to node $v_2$ given $\boldsymbol{y}((v_1, v_2, v_3))$ is $1.01$, and the only way to achieve that is by assigning flow $f_3$ to node $v_2$. If we try to assign portions of flows $f_2$ and $f_3$, we will always end up with total less than $1.01$. The reason for this is that $\delta_{f_3}^{r_1} \times \lambda_{f_3} = 10$ and if we want to replace a unit of flow $f_3$ with a unit of flow $f_2$, we will not be able to do so because flow $f_2$ is more expensive than flow $f_3$ (i.e., $\delta_{f_2}^{r_1} > \delta_{f_3}^{r_1}$). We update $\boldsymbol{y}((v_1, v_2, v_3))$ to become $[0.02, 1.01, 0]$. Finally, we solve Problem \eqref{eq:nodeCapacityAllocation} for node $v_3$ given $\boldsymbol{y}((v_1, v_2, v_3))$. In order to satisfy the equality constraints of nodes $v_1$ and $v_2$, we have to assign flow $f_1$ to node $v_1$ and flow $f_3$ to node $v_2$, as explained before. Therefore, we will not be able to assign any traffic to node $v_3$. In this case, the value of $R_4((v_1, v_2, v_3))$ is $1.03$. 

\bibliographystyle{IEEEtran}
\bibliography{references}

\end{document}